\documentclass{article}
\usepackage{amsmath}
\usepackage{amsthm}
\usepackage{amsfonts}
\usepackage{amssymb}
\usepackage{stmaryrd}
\usepackage{mathrsfs}
\usepackage{bbm}
\usepackage[all]{xy} 
\usepackage{enumerate}
\usepackage{multicol}
\usepackage{hyperref}
\usepackage{MnSymbol}
\usepackage[affil-it]{authblk}

\theoremstyle{plain}
\newtheorem{theorem}{Theorem}
\newtheorem{lemma}[theorem]{Lemma}
\newtheorem{proposition}[theorem]{Proposition}

\newtheorem{example}[theorem]{Example}

\newcommand{\op}{^{\mathrm{op}}}
\newcommand{\DL}{\mathbf{DL}}
\newcommand{\BDL}{\mathbf{BDL}}

\newcommand{\BA}{\mathbf{BA}}
\newcommand{\RK}{\mathscr{C}}
\newcommand{\SD}{\mathscr{D}}
\newcommand{\Set}{\mathbf{Set}}
\newcommand{\Pos}{\mathbf{Pos}}

\newcommand{\Upsets}{\mathcal{U}}

\newcommand{\pf}{\mathsf{Pf}}
\newcommand{\uf}{\mathsf{Uf}}
\newcommand{\Pow}{\mathsf{P}_c}
\newcommand{\Forg}{\mathsf{U}}
\newcommand{\Free}{\mathsf{F}}

\newcommand{\Id}{\mathsf{Id}}
\newcommand{\two}{\mathbbm{2}}

\newcommand{\inv}{^{-1}}
\newcommand{\epi}{\twoheadrightarrow}
\newcommand{\ce}{^\sigma}
\newcommand{\ced}{^\pi}
\newcommand{\edown}{^{\downarrow}}
\newcommand{\eup}{^{\uparrow}}
\newcommand{\restrict}{\upharpoonright}
\newcommand{\xto}{\xrightarrow}
\newcommand{\dom}{\mathrm{dom}}

\newcommand{\TRL}{T_{\mathrm{RL}}}
\newcommand{\TML}{T_{\mathrm{ML}}}
\newcommand{\LRL}{L_{\mathrm{RL}}}
\newcommand{\LML}{L_{\mathrm{ML}}}
\newcommand{\semML}{\delta^{\mathrm{ML}}}
\newcommand{\semRL}{\delta^{\mathrm{RL}}}
\newcommand{\semRLd}{\delta^{\mathrm{RL^{\partial}}}}
\newcommand{\dia}{\lozenge}
\newcommand{\ari}{\mathrm{ar}}
\newcommand{\lang}{\mathcal{L}}

\newcommand{\fus}{\otimes}
\newcommand{\lRes}{\hspace{2pt}\text{\textbf{--}}\hspace{-3pt}\fus}
\newcommand{\rRes}{\fus\hspace{-3pt}\text{\textbf{--}}\hspace{2pt}}
\newcommand{\ulRes}{-\hspace{-1pt}\fus}
\newcommand{\urRes}{\fus\hspace{-1pt}-}
\newcommand{\dual}{^{\partial}}
\newcommand{\fusd}{\oplus}
\newcommand{\lResd}{\hspace{2pt}\text{\textbf{--}}\hspace{-3pt}\fusd}
\newcommand{\rResd}{\fusd\hspace{-3pt}\text{\textbf{--}}\hspace{2pt}}
\newcommand{\ulResd}{-\hspace{-1pt}\fusd}
\newcommand{\urResd}{\fusd\hspace{-1pt}-}

\begin{document}

\title{Coalgebraic completeness-via-canonicity for \\ distributive substructural logics}

\author{Fredrik Dahlqvist\thanks{Electronic address: \href{mailto:f.dahlqvist@ucl.ac.uk}{f.dahlqvist@ucl.ac.uk}; Corresponding author}\hspace{1ex}}

\author{David Pym\thanks{Electronic address: \href{mailto:d.pym@ucl.ac.uk}{d.pym@ucl.ac.uk}}}
\affil{University College London}
\date{\vspace{-5ex}}

\maketitle
%
%
%

\begin{abstract}
We prove strong completeness of a range of substructural logics with respect to a natural poset-based relational semantics using a coalgebraic version of completeness-via-canonicity. By formalizing the problem in the language of coalgebraic logics, we develop a modular theory which covers a wide variety of different logics under a single framework, and lends itself to further extensions. Moreover, we believe that the coalgebraic framework provides a systematic and principled way to study the relationship between \emph{resource models} on the semantics side, and \emph{substructural logics} on the syntactic side.
\end{abstract}

\section{Introduction} \label{sec:introduction}
This work lies at the intersection of resource semantics/modelling, substructural logics, and the theory of canonical extensions and canonicity. These three areas respectively correspond to the semantic, proof-theoretic, and algebraic sides of the problem we tackle: to give a systematic, modular account of the relation between resource semantics and logical structure. Our approach will mostly be semantically driven, guided by the resource models of separation logic. We will therefore not delve into the proof theory of substructural logics, but rather deal with the equivalent algebraic formulations in terms of residuated lattices (\cite{2003:onoresiduated} and \cite{2007:galatosresiduated} give an overview of the correspondence between classes of residuated lattices and substructural logics).

\textbf{\emph{Resource semantics and modelling}}. Resource interpretations of substructural logics --- see, for example, \cite{Girard87,1999:BI,POY,GMP05,CP2009} --- are well-known and exemplified in the context of program verification and semantics by Ishtiaq and O'Hearn's pointer logic \cite{2001:biIshtiaq} and Reynolds' separation logic 
\cite{2002:reynoldsseparation}, each of which amounts to a model of a specific theory in Boolean BI. Resource semantics and modelling with resources has become an active field of investigation in itself (see, for example, \cite{2012:collinsonDiscipline}). Certain requirements, discussed below, seem natural (and useful in practice) in order to model naturally arising examples of resource. 
\begin{enumerate}
\item We need to be able to compare at least some resources. Indeed, in a completely discrete model of resource (i.e., where no two resources are comparable) it is impossible to model key concepts such as `having enough resources'. On the other hand, there is no reason to assume that \emph{any two} resources be comparable (e.g., heaps). This suggests at least a preorder structure on models. In fact, we take the view that comparing two resources is fundamental and, in particular, if two resources cannot be distinguished in this way then they can be identified. We thus add antisymmetry and work with posets. 
\item We need to be able to combine (some) resources to form new resources (e.g., union of heaps with disjoint domains \cite{2001:biIshtiaq}). We denote the combination operation by $\fus$. An alternative, relational, point of view is that we should be able to specify how resources can be `split up' into pairs of constituent resources. From this perspective, it makes sense to be able to list for a given resource $r$, the pairs $(s_1,s_2)$ of resources which combine to form a resource $s_1\fus s_2\leq r$.
\item All reasonable examples of resources possess `unit' resources with respect to the combination operation $\fus$; that is, special resources that leave other resources unchanged under the combination operation. 
\item The last requirement is crucial, but slightly less intuitive. In the most well-behaved examples of resource models (e.g., heaps or $\mathbb{N}$), if we are given a resource $r$ and a `part' $s$ of $r$, there exists a resource $s'$ that `completes' $s$ to make $r$; that is, we can find a resource $s'$ such that $s\fus s'=r$. More generally, given two resources $r,s$, we want to be able to find the the best $s'$ such that $s\fus s'\leq r$. In a model of resource without this feature, it is impossible to provide an answer to legitimate questions such as `how much additional resource is needed to make statement $\phi$ hold?'. Mathematically, this requirement says that the resource 
composition is a residuated mapping in both its arguments.
\end{enumerate}
The literature on resource modelling, and on separation logic in particular, is vast, but two publications 
-- \cite{2007:calcagnoPOPL} and \cite{2014:Brotherston-Villard} -- are strongly related to this work. Both 
show completeness of `resource logics' by using Sahlqvist formulas, which amounts to using 
completeness-via-canonicity (\cite{2001:ModalLogic,1994:Jonsson}).

\textbf{\emph{Completeness-via-canonicity and substructural logics}.} The logical side of resource modelling is the 
world of substructural logics, such as BI, and of their algebraic formulations; that is, residuated lattices, residuated 
monoids, and related structures. The past decade has seen a fair amount of research into proving the 
completeness of relational semantics for these logics (for BI, for example, \cite{POY,GMP05}), 
using, among other approaches, techniques from the duality theory of lattices. In \cite{2005:GerhkeSubstruct}, 
Dunn \textit{et al.} prove completeness of the full Lambek calculus and several other well-known substructural 
logics with respect to a special type of Kripke semantics by using duality theory. This type of Kripke semantics, 
which is two-sorted in the non-distributive case, was studied in detail by Gehrke in \cite{2006:GerhkeSubstruct}. 
The same techniques have been applied to prove Kripke completeness of fragments of linear logic in 
\cite{2011:coumansrelational}. Finally, the work of Suzuki \cite{2011:Suzuki} explores in much detail 
completeness-via-canonicity for substructural logics. Our work follows in the same vein but with 
with some important differences. Firstly, we use a dual adjunction rather than a dual equivalence to connect 
syntax and semantics. This is akin to working with Kripke frames rather than descriptive general frames in modal 
logics: the models are simpler and more intuitive, but the tightness of the fit between syntax and semantics is not 
as strong. Secondly, we use the topological approach to canonicity of \cite{2001:GehrkeHarding,2004:GehrkeJonsson,2006:VenemaAC} 
because we feel it is the most flexible and modular approach to building canonical (in)equations. Thirdly, we 
only consider distributive structures. This is to some extent a matter a taste. Our choice is driven by the 
desire to keep the theory relatively simple (the non-distributive case is more involved), by the fact that, from 
a resource-modelling perspective, the non-distributive case does not seem to occur `in the wild' and, finally,  
because we place ourselves in the framework of coalgebraic logic, where the category of distributive lattices 
forms a particularly nice `base category'.

\textbf{\emph{Completeness-via-canonicity, coalgebraically.}} The coalgebraic perspective brings many advantages 
to the study of completeness-via-canonicity. First, it greatly clarifies the connection between canonicity as an algebraic 
method and the existence of `canonical models'; that is, strong completeness. Second, it provides a generic framework 
in which to prove completeness-via-canonicity for a vast range of logics (\cite{2013:self}). Third, it is intrinsically 
modular; that is, it provides theorems about complicated logics by combining results for simpler ones 
(\cite{2007:CirsteaModular,2011:self}). We will return to the advantages of working coalgebraically throughout 
the paper.

\section{A coalgebraic perspective on substructural logics} \label{sec:co-alg}

We use the `abstract' version of coalgebraic logic developed in, for example, \cite{2004:KKP}, \cite{2005:KKP} 
and \cite{2010:JacobsExemplaric}; that is, we require the following basic situation:
\begin{equation}\label{diag:fundamentalSituation}
\xymatrix
{
\RK\ar@/^1pc/[rr]^{F} \ar@(l,u)^{L}& \perp&\SD\op\ar@/^1pc/[ll]^{G}\ar@(r,u)_{T\op}
}
\end{equation}
The left hand-side of the diagram is the syntactic side, and the right-hand side the semantic one. The category $\RK$ 
represents a choice of `reasoning kernel'; that is, of logical operations which we consider to be fundamental, whilst 
$L$ is a syntax constructing functor which builds terms over the reasoning kernel. Objects in $\mathscr{D}$ are the 
carriers of models and $T$ specifies the coalgebras on these carriers in which the operations defined by $L$ are  
interpreted. The functors $F$ and $G$ relate the syntax and the semantics, and $F$ is left adjoint to $G$. We will 
denote such an adjunction by $F\dashv G: \RK\to\SD$. Note, as mentioned in the introduction, that we only need 
a dual adjunction, not a full duality. 

\subsection{Syntax}

\textbf{\emph{Reasoning kernels.}} There are three choices for the category $\RK$ which are particularly suited 
to our purpose, the category $\DL$ of distributive lattices, the category $\BDL$ of bounded distributive lattices, 
and the category $\BA$ of boolean algebras. The categories $\DL, \BDL$ and $\BA$ have a very nice technical 
feature from the perspective of coalgebraic logic: each category is locally finite; that is, finitely generated objects 
are finite. This is a very desirable technical property for 
the presentation of endofunctors on this category and for coalgebraic strong completeness theorems. We denote 
by $\Free\dashv\Forg$ the usual free-forgetful adjunction between $\DL$ (resp. $\BDL$, resp. $\BA$) and $\Set$.

\textbf{\emph{True and false.}} The choice of including (or not)  $\top$ and $\bot$ to the logic is clearly provided by the choice 
of reasoning kernel.

\textbf{\emph{Algebras.}} Recall that an algebra for an endofunctor $L:\RK\to\RK$ is an object $A$ of $\RK$ together with a 
morphism $\alpha: LA\to A$. We refer to endofunctors $L:\RK\to\RK$ as \emph{syntax constructors}.

\textbf{\emph{Resource operations}.} The operations on resources specified in the introduction; that is, a combination operation 
and its left and right residuals, are introduced via the following syntax constructor:
\[
\LRL: \RK\to\RK, \begin{cases}
\LRL A =\Free\{I, a\fus b, a\lRes b, a\rRes b\mid a,b\in \Forg A\}/ \equiv \\
\LRL f: \LRL A\to \LRL B, [a]_{\equiv}\mapsto [f(a)]_{\equiv} \, , 
\end{cases}
\]
where $\equiv$ is the fully invariant equivalence relation in $\RK$ generated by the following Distribution Laws for 
non-empty finite subsets $X$ of $A$:
\begin{multicols}{2}
\begin{enumerate}[DL1.]
\item$\bigvee X\fus a=\bigvee[X\fus a]\label{ax:distribLaw1} $ 
\item$a\fus \bigvee X=\bigvee[a\fus X]$
\label{ax:distribLaw2} 
\item$a\lRes\bigwedge X=\bigwedge[a\lRes X]$
\label{ax:distribLaw3} 
\item$\bigvee X\lRes a=\bigwedge[X\lRes a]$
\label{ax:distribLaw4} 
\item$\bigwedge X\rRes a=\bigwedge [X\rRes a]$
\label{ax:distribLaw5} 
\item$a\rRes \bigvee X=\bigwedge[a\rRes X]$.
\label{ax:distribLaw6}
\end{enumerate}
\end{multicols}
\noindent where $\bigvee[X\fus a]=\bigvee\{x\fus a\mid x\in X\}$ and similarly for the other operations. For the categories $\BDL$ and $\BA$ we allow $X$ to be the empty set and use the usual convention that $\bigvee\emptyset=\bot$ and $\bigwedge\emptyset=\top$. The language defined by $\LRL$ is the free $\LRL$-algebra over $\Free V$ which we shall denote by $\lang(\LRL, V)$ or simply $\lang(\LRL)$ when a choice of propositional variables $V$ has been established. It is not difficult to see that $\lang(\LRL)$ is the language of the distributive full Lambek calculus (or residuated lattices) quotiented under the axioms of $\RK$ and DL\ref{ax:distribLaw1}-\ref{ax:distribLaw6}. An $\LRL$-algebra is simply an object of $\RK$ endowed with a nullary operation $I$ and binary operations $\fus,\lRes$ and $\rRes$ satisfying the 
distribution laws above. Note that an $\LRL$-algebra is \emph{not} a distributive residuated lattice. Only some features of 
this structure have been captured by the axioms above. But several are still missing, and will be added subsequently as 
\emph{canonical frame conditions}. $\LRL$-algebras are an example of \emph{Distributive Lattice Expansions}, or DLEs; that is, distributive lattices endowed with a collection of maps of finite arities. When $\RK=\BA$, $\LRL$-algebras are an example \emph{Boolean Algebra Expansions}, or BAEs.

\medskip 

\textbf{\emph{Modularity}.} The syntax developed above is completely modular in two respects. First, it is modular in the choice of `reasoning kernel' since the same formal functor can be overloaded to be used on several different choices of base categories. Second, and most importantly, it allows for a very concise definition and construction of the \emph{fusion} of logics (\cite{2011:self}); that is, the free combination of two logics defined on the same base categories. If $L_1,L_2:\RK\to\RK$ are functors defining languages $\lang(L_1),\lang(L_2)$, then the fusion $\lang(L_1)\oplus \lang(L_2)$ of these languages is simply given by $\lang(L_1+L_2)$ where $+$ is the object-wise coproduct in $\RK$. As an example, consider modal substructural logics; for instance, the `relevant modal logic' of \cite{2011:Suzuki} or the modal resource logics of \cite{2012:collinsonDiscipline,2013:courtaultmodal,2015:courtaultmodal}. The language of positive modal logic (\cite{1995:DunnPositiveML}) is given by the functor
\[
\LML: \RK\to\RK, \begin{cases}
\LML A =\Free\{\dia a, \square a\mid a\in \Forg A\}/\equiv\\
\LML f: \LML A\to \LML B, [a]_{\equiv}\mapsto [f(a)]_{\equiv} \, , 
\end{cases}
\]
where $\equiv$ is the fully invariant equivalence relation in $\RK$ generated by the following Distribution Laws for finite subsets $X$ of $A$:
\begin{multicols}{2}
\begin{enumerate}[ML1.]
\item$\dia (\bigvee X)=\bigvee[\dia X]\label{ax:ml:distribLaw1} $ 
\item$\square (\bigwedge X)=\bigwedge[\square X]$.
\label{ax:ml:distribLaw2} 
\end{enumerate}
\end{multicols}
\noindent where $\bigvee[\dia X]=\bigvee\{\dia x\mid x\in X\}$ and similarly for $\square$. When $\RK=\BA$ one can of course use a single modality and define its dual in the usual fashion, but nothing is lost by considering the full signature, so we will consider $\LML$ to be the functor defining modal logics across all our reasoning kernels, and the \emph{language} of modal logics is then given in our framework by the free $\LML$-algebra over $\Free V$; that is, $\lang(\LML)$. The language of the various substructural modal logics cited above, which is the fusion $\lang(\LRL)\oplus\lang(\LML)$, is thus simply given by $\lang(\LRL+\LML)$. Similarly, we can consider bi-substructural languages as is done in \cite{2015:LayeredGraphLogic}; that is, languages which allow resources to be combined in two different ways. In this case the language is simply given by $\lang(\LRL+\LRL)$.

\subsection{Coalgebraic semantics}

\textbf{\emph{Semantic domain}.} As we mentioned in the introduction, it is reasonable to assume that a model of resources should 
be a poset, and thus taking $\SD=\Pos$ is intuitively justified. This is a particularly attractive choice of `semantic domain' given 
that the category $\Pos$ is related to $\DL$ by the dual adjunction $\pf\dashv\Upsets: \DL\to\Pos\op$, where $\pf$ is the 
functor sending a distributive lattice to its poset of prime filters, and $\DL$-morphisms to their inverse images, and $\Upsets$ 
is the functor sending a poset to the distributive lattice of its upsets and monotone maps to their inverse images. When a distributive lattice is a boolean algebra, it is well-known that prime filters are maximal (i.e., ultrafilters) and the partial 
order on the set of ultrafilter is thus discrete; that is, ultrafilters are only related to themselves. Thus the dual adjunction 
$\pf\dashv\Upsets$ becomes the well-known adjunction $\uf\dashv\mathcal{P}$ between $\BA$ and $\Set\op$.

\medskip

\textbf{\emph{Coalgebras}.} Recall that a coalgebra for an endofunctor $T:\SD\to \SD$, is an object $W$ of $\SD$ together with a 
morphism $\gamma: W\to TW$. The endofunctors that we will consider are built from products and `powersets' and will be 
referred to as \emph{model constructors}. Note that $\Pos$ has products, which are simply the $\Set$ products with the obvious 
partial order on pairs of elements. The `powerset' functor which we will consider is the \emph{convex powerset} functor: $\Pow:\Pos\to\Pos$, sending a poset to its set of convex subsets, where a subset $U$ of a poset $(X,\leq)$ is convex if $x,z\in U$ and $x\leq y\leq z$ implies $y\in U$. The set $\Pow X$ is given a poset structure via the \emph{Egli-Milner} order (see \cite{2011:BKPV,2013:balanpositive}). Note that if $X$ is a set, it can be seen as a trivial poset where any element is only related to itself, and in this case it is not difficult to see that \emph{any} subset $U\subseteq X$ is convex. Thus in the case of sets, the convex subset functor $\Pow$ is simply the usual covariant powerset functor. It therefore makes sense to consider $\Pow$ over all our `semantic domains'.

\medskip

\textbf{\emph{Coalgebras for the resource operations}.} We define the following model constructor, which is used to interpret $I,\fus,\lRes$ 
and $\rRes$:
\[
\TRL: \SD\to\SD,\begin{cases} \TRL W=\two\times\Pow(W\times W)\times \Pow(W\op\times W)\times \Pow(W\times W\op) \\
\TRL f: \TRL W\to \TRL W', U\mapsto(\Id_{\two}\times (f\times f)^3)[U].
\end{cases}
\]
The intuition is that the first component of the structure map of a $\TRL$-coalgebra (to the poset $\two$) separates states into units 
and non-units. The second component sends each `state' $w\in W$ to the pairs of states which it `contains', the next two components 
are used to interpret $\lRes$ and $\rRes$, respectively, and will turn out to be very closely related to the second component. Note that if $\SD=\Pos$, the structure map of coalgebras are monotone, intuitively this means bigger resources can be split up in more ways.

\medskip 

\textbf{\emph{The semantic transformations}.} In the abstract flavour of coalgebraic logic, the semantics is provided by a natural transformation 
$\delta: LG \to G T\op$ called the \emph{semantic transformation}. We show below how this defines an interpretation map, but 
we first define our semantic transformation at every poset $W$ by its action on the generators of $\LRL GW$:
\begin{align*}
\semRL_W(I)&=\{t\in \TRL W \mid \pi_1(t)=0\in\two\}\\
\semRL_W(u\fus v)&=\{t\in \TRL W\mid \exists (x,y)\in \pi_2(t), x\in u, y\in v\}\\
\semRL_W(u\lRes w)&=\{t\in \TRL W\mid \forall (x,y)\in \pi_3(t), x\in u\Rightarrow y\in w\}\\
\semRL_W(w\rRes v)&=\{t\in \TRL W\mid\forall (x,y)\in \pi_4(t), x\in v\Rightarrow y\in w\}
\end{align*}
where $\pi_i, 1\leq i\leq 4$ are the usual projections maps, and $u,v\in GW$.

\begin{proposition}\label{prop:semPreservationProp}
The natural transformation $\semRL$ is well-defined.
\end{proposition}
\begin{proof} 
Let us first check that for any $u,v\in\Upsets W$, $\semRL_W(u\fus v), \semRL_W(u\lRes v)$ and $\semRL_W(u\rRes v)$ are upsets in $\TRL W$. Assume first that $t\in\semRL_W(u\fus v)$ and that $t\leq t'\in \TRL W$, we want to show that $t'\in \semRL_W(u\fus v)$ too. By definition of the partial order on $\TRL W$, we have that $\pi_2(t)\leq \pi_2(t')$ for the (component-wise) Egli-Milner order; that is, for each $(x,y)\in \pi_2(t)$ there exists $(x',y')\in \pi_2(t')$ such that $x\leq x'$ and $y\leq y'$. But by definition of $\semRL_W(u\fus v)$ we know that there exists $(x,y)\in \pi_2(t)$ such that $x\in u, y\in v$, and since $u,v$ are upsets it follows that $x'\in u, y'\in v$ and thus $t'\in \semRL_W(u\fus v)$ as desired. Assume now that $t\in\semRL_W(u\lRes v)$ and that $t\leq t'\in \TRL W$, we want to show that $t'\in\semRL_W(u\lRes v)$. To see that this is the case, take any $(x',y')\in \pi_3(t')$ and assume that $x'\in u$, we need to show that $y'\in v$. By definition of the Egli-Milner order we know that there exists $(x,y)\in \pi_3(t)$ such that $x'\leq x$ and $y\leq y'$ (note the inequality reversal due to the presence of $(-)\op$ in the definition of $\TRL$). Since $u$ is an upset, it follows that $x\in u$ and since $t \in\semRL_W(u\lRes v)$, it follows that $y\in v$, and thus $y'\in v$ as $v$ is an upset. The proof is identical for $\semRL_W(u\rRes v)$.

Let us now show that $\semRL_W$ satisfies the distributivity laws DL\ref{ax:distribLaw1}-\ref{ax:distribLaw6}. For any $u_1,u_2,v\in\Upsets W$ we have
\begin{align*}
\semRL_X(u_1\cup u_2,v)=&\{t\in \TRL W\mid \exists (x,y)\in \pi_2(t), x\in u_1\cup u_2, y\in v\}\\
=&\{t\in \TRL W\mid \exists (x,y)\in \pi_2(t), x\in u_1, y\in v\}\cup\\
&\{t\in \TRL W\mid \exists (x,y)\in \pi_2(t), x\in u_2, y\in v\}\\
=&\semRL_X(u_1, v)\cup\semRL_X(u_2,v),
\end{align*}
and the proof is clearly identical for the second argument. The meet preservation in the second argument of $\lRes$ is easy:
\begin{align*}
\semRL_X(v\lRes (u_1\cap u_2))=&\{t\in \TRL W\mid \forall (x,y)\in \pi_3(t), x\in v\Rightarrow y\in (u_1\cap u_2)\}\\
=&\{t\in \TRL W\mid \forall (x,y)\in \pi_3(t), x\in v\Rightarrow y\in u_1\}\hspace{2pt}\cap \\
&\{t\in \TRL W\mid \forall (x,y)\in \pi_3(t), x\in v\Rightarrow y\in u_2\}\\
=&\semRL_X(v\lRes u_1)\cap \semRL_X(v\lRes u_2). 
\end{align*}
For the anti-preservation of joins in the first argument we have
\begin{align*}
\semRL_X((u_1\cup u_2)\lRes v)=&\{t\in \TRL W\mid \forall (x,y)\in \pi_3(t), x\in (u_1\cup u_2)\Rightarrow y\in v\}\\
=&\{t\in \TRL W\mid \forall (x,y)\in \pi_3(t), x\in u_1\Rightarrow y\in v\}\hspace{2pt}\cap\\
&\{t\in \TRL W\mid \forall (x,y)\in \pi_3(t), x\in u_2\Rightarrow y\in v\}\\
=&\semRL_X(u_1\lRes v)\cap \semRL_X(u_2\lRes v), 
\end{align*}
where we use the classical equality $x\in (u_1\cup u_2)\Rightarrow y\in v$ iff $x\notin (u_1\cup u_2)$ or $y\in v$ iff ($x\notin u_1$ or $x\in v$) 
and ($x\notin u_2$ or $x\in v$). The proof 
for $\rRes$ is identical.
\end{proof}
The semantic transformations are thus well-defined. We now show how the interpretation map arises from the semantic transformation. 
Recall that, for a given syntax constructor $L:\RK\to\RK$, the language $\lang(L)$ of $L$ is the free $L$-algebra over $\Free V$. This is equivalent 
to saying that it is the initial $L(-)+\Free V$-algebra. We use initiality to define the interpretation map by putting an $L(-)+\Free V$-algebra 
structure on the `predicates' of a $T$-coalgebra $\gamma: W\to TW$; that is, on the carrier set $GW$. By definition of the coproduct, 
this means defining a morphism $LGW\to GW$ and a morphism $\Free V\to GW$. By adjointness it is easy to see that the latter is simply 
a valuation $v: V\to \Forg GW$. For the former we simply use the semantic transformation and $G$ applied to the coalgebra. The interpretation 
map $\lsem -\rsem_\gamma$ is thus given by the catamorphism:
\[
\xymatrix@C=14ex
{
L\lang(L)+\Free V\ar[dd]\ar@{-->}[r]^-{L\lsem-\rsem_\gamma+\Id_{\Free V}} & LG W+\Free V\ar[d]^{\delta_W+\Id_{\Free V}}\\
& G T W+\Free V\ar[d]^{G\gamma+v}\\
\lang(L)\ar@{-->}[r]_-{\lsem-\rsem_\gamma} & G W 
}
\]

\medskip 

\textbf{\emph{Modularity}.} Following our point on the modularity of the syntax, we highlight the modularity of the coalgebraic semantics too. Modal logic will be interpreted in $\TML$-coalgebra for the functor
\[
\TML: \SD\to\SD,\begin{cases} \TRL W=\Pow(W)\times \Pow(W) \\
\TML f: \TML W\to \TML W', U\mapsto(f)^2[U].
\end{cases}
\]
Note that we are interpreting $\dia$ and $\square$ using different relations. Modulo Dunn's \emph{interaction axioms} (\cite{1995:DunnPositiveML}) one can show that these two relations must be equal in the case of boolean modal logic, and that they can be assumed to be equal in the case of positive modal logic (although models where they are not equal will exist too). The semantics is given as usual by the transformation $\semML: \LML G\to G\TML$ defined at every poset $W$ by its action on the generators of $\LML GW$:
\begin{align*}
\semML_W(\dia u)&=\{(x,y)\in \TML W \mid x\cap u\neq\emptyset\}\\
\semML_W(\square u)& =\{(x,y)\in \TML W \mid y\subseteq u\}
\end{align*}

Model constructors and semantic transformations can be assembled in a way that is dual to the the syntax constructors; that is, using products rather than co-products. Formally, for languages defined by functors $L_1,L_2;\RK\to\RK$ interpreted in coalgebras for the functors $T_1,T_2;\SD\to\SD$ by semantic transformations $\delta^1,\delta^2$ respectively, the fusion $\lang(L_1+L_2)$ is interpreted in $T_1\times T_2$-coalgebra, where the product is taken object-wise in $\SD$, via the semantic transformation
\[
(G\pi_1+G\pi_2)\circ \delta^1+\delta^2: L_1G+L_2G\to GT_1+GT_2\to GT_1\times GT_2
\]
In particular, the semantics of the modal substructural logics defined above is given by the following interpretation maps:
\[
\xymatrix@C=11ex
{
(\LML\hspace{-2pt}+\hspace{-2pt}\LRL)\lang(\LML\hspace{-2pt}+\hspace{-2pt}\LRL(-))\hspace{-2pt}+\hspace{-2pt}\Free V\ar[dd]\ar@{-->}[r]^-{\txt{\scriptsize{$\LML+\LRL\lsem-\rsem_{(\gamma_1\times \gamma_2)}+\Id_{\Free V}$}\\\hspace{1pt}}}& \LML G W\hspace{-2pt}+\hspace{-2pt}\LRL GW\hspace{-2pt}+\hspace{-2pt}\Free V\ar[d]_{\semML_W+\semRL_W+\Id_{\Free V}}\\
& G \TML W\hspace{-2pt}+\hspace{-2pt}G\TRL W\hspace{-2pt}+\hspace{-2pt}\Free V\ar[d]_{G(\gamma_1\times \gamma_2)\circ (G\pi_1+G\pi_2)+v}\\
\lang(\LML\hspace{-2pt}+\hspace{-2pt}\LRL)\ar@{-->}[r]_{\lsem-\rsem_{(\gamma_1\times \gamma_2)}} & G W 
}
\]

\subsection{Advantages of the coalgebraic approach.}
Before we move on to the technical part of this paper, we return to the advantages of our set up. From the perspective of studying the relation between resource semantics and logics, the fundamental situation described by Diagram \ref{diag:fundamentalSituation} is particularly promising. Going from substructural logics to resource models, the coalgebraic approach allows us to `guess' and generate appropriate resource models. Indeed, starting from a `reasoning kernel' $\RK$, the existence of a dual adjunction $F\dashv G$ with a category $\SD$ restricts the kind of model carriers we should consider. Moreover, as we will see later, objects of the type $GFA$ for $A\in\RK$ will play a crucial role and should be \emph{canonical extensions}. This extra requirement determines to a great extent the useful structure(s) one ought to consider for the carriers of resource models. For example when $\RK=\DL$, we cannot take $\SD=\Set$, because $GFA$ is then given by $\mathcal{P}\pf A$ which is not the canonical extension of $A$. It is therefore the framework itself which suggests that non-boolean substructural logics should have posets of resources as their models. Similarly, as we have shown above, the choice of $T$, that is to say of relational structure on the carrier, can be guessed from that of $L$ in a systematic fashion -- at least for the languages we consider here.

Conversely, if we start from requirements on resource models, such as the natural conditions listed in the introduction, we can work from resource model to logic via the existence of a dual adjunction $F\dashv G$ and the constraint that $GFA$ should be the canonical model of $A$. In this way, the `natural' logics to reason about partially ordered models of resources are positive; that is, based on $\DL$. Moreover, the relational structure suggested in the introduction suggests adding binary modalities, in other words functors $L:\DL\to\DL$ building binary `modal' formulas over $\DL$, as was done in this section. Thus we see that in either direction the categorical clarity of coalgebraic logics provides us with a natural and principled methodology for building resource models from substructural logics and vice versa.

Finally, we note that recent work on \emph{positive} coalgebraic logics (\cite{2011:BKPV,2013:balanpositive}) suggests that what is known of boolean modal logics with relational semantics can be adapted in a systematic and principled way to the positive modal logics that we are are considering here. Indeed, our choice of semantics in terms of convex powerset coalgebras is dictated by the fact that this functor is a universal extension to $\Pos$ of the usual powerset functor on $\Set$. In this sense, the coalgebraic perspective also suggests what a `correct' relational model on a poset of resources should be.
  
\section{J\'{o}nsson-Tarski extensions}

The languages $\lang(\LRL)$ and $\lang(\LML)$ which we have introduced earlier are part a class of logics with a very strong property: they  are strongly complete with respect to their semantics. This is what we will now establish, and it is the first step is showing strong completeness of more complex logics based on the languages $\lang(\LRL)$ and $\lang(\LML)$. The proof is an application of the coalgebraic J\'{o}nsson-Tarksi 
theorem.

\begin{theorem}[Coalgebraic J\'{o}nsson-Tarksi theorem, \cite{2005:KKP}]\label{thm:JonTar}
Assuming the basic situation of Diagram (\ref{diag:fundamentalSituation}) and a semantic transformation $\delta: LG\to GT$, if its adjoint 
transpose $\hat{\delta}: TF\to FL$ has a right-inverse $\zeta: FL\to TF$, then for every $L$-algebra $\alpha: LA\to A$, the 
embedding $\eta_A: A\to GFA$ of $A$ into its canonical extension can be lifted to the following $L$-algebra embedding:

\begin{equation}\label{diag:jontar}
\xymatrix
{
LA\ar[rrr]^{\alpha}\ar[d]_{L\eta_A} & & & A\ar[d]^{\eta_A}\\
LGF A\ar[r]_{\delta_{FA}} & GTFA\ar[r]_{G\zeta_A} & GFLA\ar[r]_{GF\alpha} & GFA
}
\end{equation}
\end{theorem}

We call the coalgebra $\zeta\circ F\alpha: FA\to TFA$ a \emph{canonical model} of (the $L$-algebra) $A$. If $A$ is the free $L$-algebra over $\Free V$ we recover the usual notion of canonical model. The `truth lemma' follows from the definition of $\eta$. We will call the $L$-algebra $LGFA\to GFA$ defined by Diagram (\ref{diag:jontar}) a \emph{J\'{o}nsson-Tarski extension} of the $\alpha: LA\to A$.  

We now prove the existence of canonical models for the logics defined by $\LML$ and $\LRL$. The result generalizes Lemma 5.1 of \cite{1995:DunnPositiveML}, which builds canonical models for countable DLs with a unary operator, and Lemma 4.26 of \cite{2001:ModalLogic}, which builds canonical models for countable BAs with $n$-ary operators. We essentially show how to build canonical models for arbitrary DLs with $n$-ary expansions all of whose arguments either (1) preserve joins or anti-preserve meets, or (2) preserve meets or anti-preserve joins. The proof is rather involved and is detailed in the appendix.

\begin{theorem}\label{thm:rightInvDelta}
The adjoint transpose of the transformation $\semRL: \LRL G\to G\TRL$ (resp. $\semML: \LML G\to G\TML$) has right inverses at every distributive lattice.
\end{theorem}

\medskip

\textbf{\emph{Strong completeness}:} Let us now define what we exactly mean by strong completeness.  Let $\RK$ be $\DL,\BDL$ or $\BA$, $L:\RK\to\RK$, $V$ be a set of propositional variables, $q:\lang(L)\epi \mathcal{Q}$ be a regular epi, and let $\Phi,\Psi\subseteq \mathcal{Q}$ be two families of `formulas' such that $\Phi\not\vdash\Psi$; that is, 
such that no finite set $\Phi_0$ of elements of $\Phi$ and no finite set $\Psi_0$ of elements of $\Psi$ can be found such that $\bigwedge\Phi_0\leq \bigvee\Psi_0$. The statement that $\mathcal{Q}$ is strongly complete w.r.t. to a class $\mathcal{T}$ of $T$-coalgebras means that for any such choice of $\Phi,\Psi$ there exists a $T$-coalgebra  
$\gamma: X\to TX$ in $\mathcal{T}$, a valuation $v:\Free V\to GX$, and a point $x\in X$ such that $x\in\lsem a\rsem_{(\gamma,v)}$ for all $a\in\Phi$ and $x\notin\lsem b\rsem_{(\gamma,v)}$ for all $b\in \Psi$. 

\begin{theorem}\label{thm:strongComp}
The logic defined by $\LRL$ (resp. $\LML$) is sound and strongly complete with respect to the class of all $\TRL$- (resp. $\TML$-) coalgebras.
\end{theorem}
\begin{proof}
We show the case of $\lang(\LRL)$. Let $\Phi,\Psi\subseteq \lang(\LRL)$ and $\Phi\not\vdash\Psi$. Then the filter $\langle\Phi\rangle\eup$ generated by $\Phi$ and the ideal $\langle\Psi\rangle\edown$ generated by $\Psi$ obey $\langle\Phi\rangle\eup\cap \langle\Psi\rangle\edown=\emptyset$. By the PIT there exists a prime filter $w_{\Phi}$ extending $\langle\Phi\rangle\eup$ such that $w_\Phi\cap\langle\Psi\rangle\edown=\emptyset$. By Theorems \ref{thm:JonTar} and \ref{thm:rightInvDelta}, the $L$-algebra $\lang(\LRL)$ has a J\'{o}nsson-Tarski extension which provides an interpretation in the $T$-coalgebra $\pf\lang(\LRL)\to \pf \LRL\lang(\LRL)\xto{\zeta_{\lang(\LRL)}} T\pf\lang(\LRL)$ coinciding with $\eta_{\lang(\LRL)}$. In this interpretation $w_\Phi\in\lsem a\rsem$ for all $a\in\Phi$ and $w_\Phi\notin\lsem b\rsem$ for all $b\in\Psi$.
\end{proof}

\section{Canonical extensions and canonical equations}

In the previous section we have shown how to embed an $L$-algebra with carrier $A$ into an $L$-algebra with carrier $GFA$. When $\RK=\DL,\BDL$ or $\BA$ carriers of this shape are known as \emph{canonical extensions} (and denoted $A\ce$) and a great deal is known about them. The theory of canonical extensions in $\DL$ has been extended to boolean algebras with operators (BAOs) (\cite{1951:Jonsson}) and to distributive lattice expansions (DLEs) (\cite{1994:GehrkeJonsson,2004:GehrkeJonsson}) and forms the basis of the theory of canonicity which consists in determining when the validity of an an equation  in a DLEs transfer to its canonical extension; that is, when $A\models s=t$ implies $A\ce\models s=t$. Note that the canonical and J\'{o}nsson-Tarski extensions are in general not equal. This section deals only with canonical extensions, but we will see in the next section how these results can be combined with the J\'{o}nsson-Tarski construction of Theorem \ref{thm:JonTar}.

\subsection{Canonical extension of distributive lattices}

We now briefly present the salient facts about canonical extensions for distributive lattices. For any $A$ in $\DL$, $\Upsets\pf A$ is known as the \emph{canonical extension} of $A$ and denoted $A\ce$. It can be characterised uniquely up to isomorphism through purely algebraic properties, namely that $A$ is \emph{dense} and \emph{compact} in $A\ce$. For our purpose however, \emph{defining} the canonical extension of $A$ as $\Upsets\pf A$ will be sufficient. The canonical extension $A\ce$ of a distributive lattice $A$ is always \emph{completely distributive} (see \cite{2004:GehrkeJonsson}). The following terminology will be important: $A\ce$ is a completion of $A$ and all joins of elements of $A$ therefore exist in $A\ce$, such elements are called \emph{open} and their set is denoted by $O(A)$. Dually, meets in $A\ce$ of elements of $A$ will be called \emph{closed} and their set denoted $K(A)$. Elements of $A=K(A)\cap O(A)$ are therefore called \emph{clopens}.

\subsection{Canonical extension of distributive lattice expansions}\label{sec:DLE}

We now sketch the theory of canonical extensions for Distributive Lattice Expansions (DLE) --- for the details, 
see \cite{1994:GehrkeJonsson,2004:GehrkeJonsson}. Each map $f: \Forg A^n\to \Forg A$ can be extended to a map 
$(\Forg A\ce)^n\to\Forg A\ce$ in two canonical ways: 
\begin{align*}
f\ce(x)=\bigvee\{\bigwedge f[d,u]\mid K^n\ni d\leq x\leq u\in O^n\}\\
f\ced(x)=\bigwedge\{\bigvee f[d,u]\mid K^n\ni d\leq x\leq u\in O^n\}, 
\end{align*}
where $f[d,u]=\{f(a)\mid a\in A^n, d\leq a\leq u\}$. Note that since $A$ is compact in $A\ce$ the intervals $[d,u]$ are never empty, which 
justifies these definitions. For a signature $\Sigma$, the \emph{canonical extension} of a 
$\Sigma$-DLE $(A,(f_s: \Forg A^{\ari(n)}\to \Forg A)_{s\in \Sigma})$ is defined to be the $\Sigma$-DLE 
$(A\ce, (f_s\ce: \Forg (A\ce)^{\ari(n)}\to \Forg A\ce)_{s\in \Sigma})$, and similarly for BAEs. We summarize some 
important facts about canonical extensions of maps in the following proposition, proofs can be found in, for example, 
\cite{2001:GehrkeHarding,2004:GehrkeJonsson,2006:VenemaAC}: 

\begin{proposition}
Let $A$ be a distributive lattice, and $f: \Forg A^n\to \Forg A$. 
\begin{enumerate}
\item $f\ce\restrict A^n=f\ced\restrict A^n=f$.
\item $f\ce\leq f\ced$ under pointwise ordering.
\item If $f$ is monotone in each argument, then
$f\ce\restrict(K\cup O)^n=f\ced\restrict(K\cup O)^n$.
\end{enumerate}
\end{proposition}

We call a monotone map $f: \Forg A^n\to \Forg A$ \emph{smooth in its $i^{th}$ argument} ($1\leq i\leq n$) if, for every $x_1,\ldots,x_{i-1},x_{i+1},\ldots,x_n\in K\cup O$,
\[
f\ce(x_1,\ldots,x_{i-1},x_i,x_{i+1},\ldots,x_n)=f\ced(x_1,\ldots,x_{i-1},x_i,x_{i+1},\ldots,x_n), 
\]
for every $x_i\in A\ce$. A map $f: \Forg A^n\to \Forg A$ is called \emph{smooth} if it is smooth in each of its arguments.

In order to study effectively the canonical extension of maps, we need to define six topologies on $A\ce$. First, we define $\sigma\eup,\sigma\edown$ and $\sigma$ as the topologies generated by the bases $\{\uparrow p\mid p\in K\},\{\downarrow u\mid u\in O\}$ and $\{\uparrow p\cap\downarrow u\mid K\ni p\leq \in O\}$. The next set of topologies 
is well-known to domain theorists: a \emph{Scott open}  in $A\ce$ is a subset $U\subseteq A\ce$ such that (1) $U$ is an upset and (2) 
for any up-directed set $D$ such that $\bigvee D\in U$, $D\cap U\neq\emptyset$. The collection of Scott opens forms a topology 
called the \emph{Scott topology}, which we denote $\gamma\eup$. The dual topology will be denoted by $\gamma\edown$, and their 
join by $\gamma$. It is not too hard to show (see \cite{2001:GehrkeHarding,2006:VenemaAC}) that 
$\gamma\eup\subseteq \sigma\eup$, $\gamma\edown\subseteq \sigma\edown$, and $\gamma\subseteq\sigma$.
We denote the product of topologies by $\times$, and the $n$-fold product of a topology $\tau$ by $\tau^n$. The following result 
shows why these topologies are important: they essentially characterize the canonical extensions of maps:  

\begin{proposition}[\cite{2001:GehrkeHarding}]\label{prop:canExtTopology}
For any DL $A$ and any map $f: \Forg A^n\to \Forg A$, 
\begin{enumerate}
\item $f\ce$ is the largest $(\sigma^n,\gamma\eup)$-continuous extension of $f$, 
\item $f\ced$ is the smallest $(\sigma^n,\gamma\edown)$-continuous extension of $f$
\item $f$ is smooth iff it has a unique $(\sigma^n,\gamma)$-continuous extension.
\end{enumerate} 
\end{proposition}

From this important result, it is not hard to get the following key theorem, sometimes known as \emph{Principle of Matching Topologies}, 
which underlies the basic `algorithm' for canonicity: 

\begin{theorem}[Principle of Matching Topologies,\cite{2001:GehrkeHarding,2006:VenemaAC}]\label{thm:PrincMatchTopol}
Let $A$ be a distributive lattice, and $f: \Forg A^n\to \Forg A$ and $g_i:\Forg A^{m_i}\to \Forg A, 1\leq i\leq n$ be arbitrary maps. Assume 
that there exist topologies $\tau_i$ on $A$, $1\leq i\leq n$ such that each $g_i\ce$ is $(\sigma^{m_i},\tau_i)$-continuous, then
\begin{enumerate}
\item \hspace{-3pt}if $f\ce$ is $(\tau_1\times\ldots\times\tau_n,\gamma\eup)$-continuous, then $
f\ce(g_1\ce,\ldots,g_n\ce)\leq (f(g_1,\ldots,g_n))\ce$
\item \hspace{-3pt}if $f\ce$ is $(\tau_1\times\ldots\times\tau_n,\gamma\edown)$-continuous, then $
f\ce(g_1\ce,\ldots,g_n\ce)\geq (f(g_1,\ldots,g_n))\ce$
\item \hspace{-3pt}if $f\ce$ is $(\tau_1\times\ldots\times\tau_n,\gamma)$-continuous, then $f\ce(g_1\ce,\ldots,g_n\ce)=(f(g_1,\ldots,g_n))\ce $.
\end{enumerate}
\end{theorem}

The last piece of information we need to effectively use the Principle of Matching Topologies is to determine when maps are 
continuous for a certain topology, based on the distributivity laws they satisfy. For our purpose the following results will be sufficient:

\begin{proposition}[\cite{1994:GehrkeJonsson,2001:GehrkeHarding,2004:GehrkeJonsson,2006:VenemaAC}]\label{prop:presPropAndTopol}
Let $A$ be a distributive lattice, and let $f: \Forg A^n\to \Forg A$ be a map. For every ($n-1)$-tuple $(a_i)_{1\leq i\leq n-1}$, we denote 
by $f_a^k: A\to A$ the map defined by $x\mapsto f(a_1,\ldots,a_{k-1},x,a_{k},\ldots,a_{n-1})$.
\begin{enumerate}
\item If $f_a^k$ preserves binary joins, then $(f\ce)_a^k$ preserve all non-empty joins and is $(\sigma\edown,\sigma\edown)$-continuous.
\item If $f_a^k$ preserves binary meets, then $(f\ce)_a^k$ preserve all non-empty meets and is $(\sigma\eup,\sigma\eup)$-continuous.
\item If $f_a^k$ anti-preserves binary joins (i.e., turns them into meets), then $(f\ce)_a^k$ anti-preserve all non-empty joins and is $(\sigma\edown,\sigma\eup)$-continuous.
\item If $f_a^k$ anti-preserves binary meets (i.e., turns them into joins), then $(f\ce)_a^k$ anti-preserve all non-empty meets and is $(\sigma\eup,\sigma\edown)$-continuous.
\item In each case $f$ is is smooth in its $k^{th}$ argument.
\end{enumerate}
\end{proposition}

\subsection{Canonical (in)equations}

To say anything about the canonicity of equations, we need to compare interpretations in $A$ with interpretations in $A\ce$. It is natural to 
try to use the extension $(\cdot)\ce$ to mediate between these interpretations, but $(\cdot)\ce$ is defined on maps, not on terms. Moreover, 
not every valuation on $A\ce$ originates from valuation on $A$. We would therefore like to recast the problem in such a way that (1) terms 
are viewed as maps, and (2) we do not need to worry about valuations.  

\medskip
 
\textbf{\emph{Term functions}.} The solution is to adopt the language of \emph{term functions} (as first suggested in \cite{1994:Jonsson}). Given a signature $\Sigma$, let $\mathsf{T}(V)$ denote 
the language of $\Sigma$-DLEs (or $\Sigma$-BAEs) over a set $V$ of propositional variables. We view each term $t\in\mathsf{T}(V)$ as 
defining, for each $\Sigma$-DLE $A$, a map $t^A: A^n\to A$. This allows us to consider its canonical extension $(t^A)\ce$, and also allows 
us to reason without having to worry about specifying valuations. Formally, given a signature $\Sigma$ and a set $V$ a propositional variables, 
we inductively define the term function associated with an element $t$ built from variables $x_1,\ldots,x_n\in V$ as follows:
\begin{itemize}
\item $x_i^A=\pi_i^n:A^n\to A, 1\leq i\leq n$; 
\item $(f(t_1,\ldots, t_m))^A=f^A\circ \langle t_1^A, \ldots, t_m^A\rangle$.
\end{itemize}
where $\pi_i$ is the usual projection on the $i^{th}$ component, $f^A$ is the interpretation of the symbol $f$ in $A$ and 
$\langle t_1^A,\ldots,t_m^A\rangle$ is usual the product of $m$ maps. Note that in this definition we work in $\Set$, and the 
building blocks of term functions are thus the variables in $V$ (interpreted as projections) and all operation symbols, including 
$\vee,\wedge$ and possibly $\neg$.
\begin{proposition}\label{prop:termFct}
Let $s,t$ be terms in the language defined by a signature $\Sigma$ and $A$ be a $\Sigma$-DLE,
\[
A\models s=t\text{ iff }s^A=t^A \, .
\]
\end{proposition}

\medskip 

\textbf{\emph{Canonical (in)equations}.} An equation $s=t$ where $s,t\in \mathsf{T}(V)$ is called \emph{canonical} if $A\models s=t$ implies 
$A\ce\models s=t$, and similarly for inequations. Following \cite{1994:Jonsson}, we say that $t\in\mathsf{T}(V)$ is \emph{stable} if $(t^A)\ce=t^{A\ce}$, that $t$ is \emph{expanding} if $(t^A)\ce\leq t^{A\ce}$, and that $t$ is \emph{contracting} if $(t^A)\ce\geq t^{A\ce}$, 
for any $A$. The inequality between maps is taken pointwise. The following proposition illustrates the usefulness of these notions: 

\begin{proposition}[\cite{1994:Jonsson}]\label{prop:canEqu}
If $s,t\in\mathsf{T}(V)$ are stable then the equation $s=t$ is canonical. Similarly, let $s,t\in\mathsf{T}(V)$ such that $s$ is contracting and $t$ is 
expanding, then the inequality $s\leq t$ is canonical. 
\end{proposition}
\begin{proof}
Let $A$ be an arbitrary $\Sigma$-DLE. If $A\models s=t$, then $s^A=t^A$ by Proposition \ref{prop:termFct}. Therefore $(s^A)\ce=(t^A)\ce$ and 
thus $s^{A\ce}=t^{A\ce}$, by stability, and it follows that $A\ce\models s=t$ by Proposition \ref{prop:termFct}.

Similarly, if $A\models s\leq t$ then $s^A\leq t^A$  by Proposition \ref{prop:termFct} and thus $(s^A)\ce\leq (t^A)\ce$. By the assumptions on 
$s$ and $t$, this means that we also have $s^{A\ce}\leq t^{A\ce}$, and thus $A\ce\models s\leq t$ by Proposition \ref{prop:termFct}.
\end{proof}

\section{Coalgebraic Completeness via-canonicity}

In this section we will combine the results of Sections 2 and 3. We will first exhibit a  set of canonical axioms which complete the definition of $\LRL$ and completely axiomatize the distributive full Lambek calculus. This will prove that the variety defined by these axioms is canonical; 
that is, closed under canonical extension. We will then show that the canonical and J\'{o}nsson-Tarski extensions defined by Theorems \ref{thm:rightInvDelta} and \ref{thm:JonTar} coincide. This will allow us to conclude strong completeness of the distributive full Lambek calculus.

\subsection{Axiomatizing distributive residuated lattices}

So far we have only captured part of the structure of distributive residuated lattices, namely we have enforced the distribution properties of $\to,\fus,\lRes$ and $\rRes$ by our definition of the syntax constructor $\LRL$. In order to capture 
the rest of the structure we now add axioms which, when added to DL\ref{ax:distribLaw1}-\ref{ax:distribLaw6}, fully axiomatize
distributive residuated lattices. Due to the constraints that these axioms must be canonical, we choose the following Frame Conditions:

\begin{multicols}{2}
\begin{enumerate}[FC1.]
\item $a\fus I=a$, $I\fus a=a$
\label{ax:FrameCond1} 
\item $I\leq a\lRes a$, $I\leq a\rRes a$
\label{ax:FrameCond2} 
\item $a\fus(b\lRes c)\leq (a\fus b)\lRes c$ \label{ax:FrameCond3} 
\item $(c\rRes b)\fus a\leq c\rRes(a\fus b)$ \label{ax:FrameCond4} 
\item $(a\rRes b)\fus b\leq a$  
\label{ax:FrameCond5} 
\item $b\fus(b\lRes a)\leq a$.
\label{ax:FrameCond6}
\end{enumerate}
\end{multicols}

\begin{proposition}\label{prop:ResLattAx}
The axioms DL\ref{ax:distribLaw1}-\ref{ax:distribLaw6} and FC\ref{ax:FrameCond1}-\ref{ax:FrameCond6} 
axiomatize distributive residuated lattices.
\end{proposition}
\begin{proof}
It is straighforward to check that axioms DL\ref{ax:distribLaw1}-\ref{ax:distribLaw6} and FC\ref{ax:FrameCond1}-\ref{ax:FrameCond6} 
hold in any residuated lattice. Conversely, we show that if FC\ref{ax:FrameCond1}-\ref{ax:FrameCond6} hold in an $\LRL$-algebra, then 
this $\LRL$-algebra is a residuated lattice. It is clear from FC\ref{ax:FrameCond1} that $\fus$ defines a monoid on the carrier set. It remains to 
show that the residuation conditions are satisfied. Assume that $a\fus b\leq c$. We have
\begin{align*}
I&\leq b\lRes b & \text{FC}\ref{ax:FrameCond2}\\
a&\leq a\fus(b\lRes b) & \text{FC}\ref{ax:FrameCond1}\text{ and monotonicity of }\fus\text{ from }\LRL\\
 & \leq (a\fus b)\lRes b & \text{FC}\ref{ax:FrameCond3}\\
& \leq c\lRes b & \text{Monotonicity of }\lRes\text{ from }\LRL
\end{align*}
Now assume that $a\leq b\lRes c$. Then we have
\begin{align*}
a\fus b&\leq (b\lRes c)\fus b & \text{Monotonicity of }\fus\text{ from }\LRL\\
& \leq c & \text{FC}\ref{ax:FrameCond5}. 
\end{align*}
The proof for the left residual $\rRes$ is identical. Note that the monotonicity of the operators are consequences of DL\ref{ax:distribLaw1}-\ref{ax:distribLaw6}.
\end{proof}

We now show one of the crucial steps.
\begin{proposition}\label{prop:FrameCondCan}
The axioms FC\ref{ax:FrameCond1}-\ref{ax:FrameCond6} are canonical.
\end{proposition}
\begin{proof}The proof is an application of Theorem \ref{thm:PrincMatchTopol} and Proposition \ref{prop:canEqu}.

FC\ref{ax:FrameCond1}: Since $\fus$ preserves binary joins in each argument, it is smooth by Prop. \ref{prop:presPropAndTopol}, 
and it follows that it is $(\sigma^2,\gamma)$-continuous. Since $\pi_1\ce$ and $I\ce$ are trivially $(\sigma,\sigma)$-continuous, it follows from Theorem \ref{thm:PrincMatchTopol} that $(\fus \circ\langle \pi_1, I\rangle)\ce=\fus\ce\circ \langle \pi_1, 1\rangle\ce$. 
Each side of the equation is thus stable and the result follows from Prop. \ref{prop:canEqu}.

FC\ref{ax:FrameCond2}: $I$ is stable and thus contracting, and $(\lRes \circ \langle\pi_1,\pi_1\rangle)\ce=\lRes\ce\circ \langle\pi_1,\pi_1\rangle\ce$, 
since $\pi_1\ce$ is $(\sigma,\sigma)-$continuous and $\lRes\ce$ is smooth. The RHS of the 
inequality is thus stable, and a fortiori expanding, and the inequality is thus canonical.

FC\ref{ax:FrameCond3}-\ref{ax:FrameCond4}:  Since $\fus\ce$ preserve joins in each argument, it preserves up-directed ones, and is thus $((\gamma\eup)^2,\gamma\eup)$-continuous. Since $\lRes\ce$ is smooth it is in particular $(\sigma^2,\gamma\eup)$-continuous. Since $\pi_1\ce$ is $(\sigma,\gamma\eup)$-continuous, we get that $\fus\ce\circ\langle\pi_1\ce,\lRes\ce\circ\langle\pi_2\ce,\pi_3\ce\rangle\rangle$ is $(\sigma^3,\gamma\eup)$-continuous and thus contracting. For the RHS, note that since $\lRes\ce$ preserves meets in its first argument, it must 
in particular preserve down-directed ones, thus $\lRes\ce$ is $(\gamma\edown,\gamma\edown)$-continuous in its first argument. Similarly, since 
$\lRes\ce$ anti-preserve joins in its second argument, it must in particular anti-preserve up-directed ones, and is thus 
$(\gamma\eup,\gamma\edown)$-continuous in its second argument. This means that $\lRes\ce$ is $(\gamma^2,\gamma\edown)$-continuous. 
We thus have that the full term is $(\sigma^3,\gamma\edown)$ continuous, and thus expanding. The inequation is therefore canonical. 

FC\ref{ax:FrameCond5}-\ref{ax:FrameCond6}: The LHS is contracting by the same reasoning as above, and the RHS is stable and thus expanding.

\end{proof}

\subsection{J\'{o}nsson-Tarski vs canonical extensions}

We have just shown that the variety of $\LRL$-algebras defined by the equations FC\ref{ax:FrameCond1}-\ref{ax:FrameCond6} is canonical;  that is, closed under canonical extension. However, since we want to exhibit models, what we really need to show is that the variety defined by  FC\ref{ax:FrameCond1}-\ref{ax:FrameCond6} is closed under \emph{J\'{o}nsson-Tarski extensions}. Fortunately, for the logics of interest to us here the two extensions in fact coincide. This is what we will now show. The proof is not difficult but rather long, and can be found in the appendix.

\begin{proposition}\label{prop:JonTarskiCanExt}
The structure map of the J\'{o}nsson-Tarski extension of an $\LRL$-algebra is equal to the canonical extension of its structure map (in the sense of Section \ref{sec:DLE}).
\end{proposition}

\subsection{Strong completeness}

We are now ready to combine all our results and to state and prove our main completeness theorem.

\begin{theorem}[Strong completeness theorem]\label{thm:strCompMain}
The Distributive Full Lambek Calculus is strongly complete with respect to the class of $\TRL$-coalgebras validating FC\ref{ax:FrameCond1}-\ref{ax:FrameCond6}.
\end{theorem}
\begin{proof}[Theorem \ref{thm:strCompMain}]
Let $\Phi,\Psi$ be (not necessarily finite) subsets of $\lang(\LRL)$; that is, elements of the free $\LRL$-algebra over $\Free V$, such that 
\[
\mathrm{FC}\ref{ax:FrameCond1}-\ref{ax:FrameCond6}+\Phi\nvdash\Psi
\]
We need to find a $\TRL$-model validating the axioms FC\ref{ax:FrameCond1}-\ref{ax:FrameCond6} such that each $a\in \Phi$ and no $b\in\Psi$ is satisfied in this model.
Now consider the Lindenbaum-Tarski $\LRL$-algebra $\alpha: \LRL\mathcal{L}\to\mathcal{L}$ defined by axioms FC\ref{ax:FrameCond1}-\ref{ax:FrameCond6}; that is,  
\[
\mathcal{L}=\lang(\LRL)/(\mathrm{FC}\ref{ax:FrameCond1}-\ref{ax:FrameCond6}),
\]
where the quotient is under the fully invariant equivalence relation in $\RK$ generated by the frame conditions FC\ref{ax:FrameCond1}-\ref{ax:FrameCond6}. Note that this algebra comes equipped with a canonical valuation $v: \Free V\to \mathcal{L}$. By construction, $\mathcal{L}$ validates FC\ref{ax:FrameCond1}-\ref{ax:FrameCond6}, and since we've established, in Proposition \ref{prop:FrameCondCan}, that they are canonical, the $\LRL$-algebra 
\[
\LRL \Upsets\pf\mathcal{L}\xto{\alpha\ce}\Upsets\pf\mathcal{L}
\]
also validates these axioms. By Proposition \ref{prop:JonTarskiCanExt} 
we know that this $\LRL$-algebra is the J\'{o}nsson-Tarski extension of $\mathcal{L}$, and as a consequence 
\[
\LRL \Upsets\pf\mathcal{L}\xto{\semRL_{\pf \mathcal{L}}}\Upsets \TRL\pf\mathcal{L}\xto{\Upsets\zeta_{\mathcal{L}}}\Upsets\pf\LRL\mathcal{L}\xto{\Upsets\pf \alpha}\Upsets\pf\mathcal{L}
\]
validates FC\ref{ax:FrameCond1}-\ref{ax:FrameCond6}. As a direct consequence of the definition of coalgebraic semantics, we have the following commutative diagram:
\[
\xymatrix@C=10ex
{
\LRL\lang(\LRL)+\Free V\ar[ddd]\ar[r]^-{\LRL\lsem-\rsem_{\mathcal{L}}+\Id_{\Free V}} \ar@/^2pc/[rr]^{\LRL\lsem-\rsem_{\Upsets\pf\mathcal{L}}+\Id_{\Free V}}
& 
\LRL\mathcal{L}+\Free V\ar[ddd]_{\alpha}\ar[r]^-{\LRL\eta_{\mathcal{L}}+\Id_{\Free V}} & 
\LRL \Upsets\pf\mathcal{L}+ \Free V\ar[d]_{\semRL_{\pf \mathcal{L}}+\Id_{\Free V}}\\
& & \Upsets \TRL\pf\mathcal{L}+\Free V\ar[d]_{\Upsets\zeta_{\pf\mathcal{L}}+\Id_{\Free V}}\\
& & \Upsets\pf\mathcal{L}+\Free V\ar[d]_{\Upsets\pf \alpha+v}\\
\lang(\LRL)\ar[r]_-{\lsem-\rsem_{\mathcal{L}}}\ar@/_2pc/[rr]_{\lsem-\rsem_{\Upsets\pf\mathcal{L}}} & \mathcal{L}\ar[r]_-{\eta_{\mathcal{L}}} & 
\Upsets\pf\mathcal{L}}  
\]
It follows easily that at every prime filter $w\in\pf\mathcal{L}$, $w\models$FC\ref{ax:FrameCond1}-\ref{ax:FrameCond6}, since $\lsem -\rsem_{\Upsets\pf\mathcal{L}}$ must factor through $\lsem-\rsem_{\mathcal{L}}$ which ensures precisely that FC\ref{ax:FrameCond1}-\ref{ax:FrameCond6} are valid. Thus $\pf\mathcal{L}$ is a model validating the axioms. We now need to find a point in $w_{\Phi}\in\pf\mathcal{L}$ such that $w\models\Phi$ but $w\not\models\Psi$. For this we start by considering the filter-ideal pair
\[
(\langle \Phi\rangle\eup, \langle\Psi\rangle\edown)
\]
where $\langle\Phi\rangle\eup$ is the filter generated by the equivalence classes in the Lindenbaum-Tarski algebra $\mathcal{L}$ of formulas in $\Phi$, and similarly for the ideal generated by $\Psi$. It is clear that $\langle\Phi\rangle\eup$ is proper, or else we would necessarily have $\mathrm{FC}\ref{ax:FrameCond1}-\ref{ax:FrameCond6}+\Phi\vdash\Psi$, a contradiction. For the same reason it is clear that $\langle \Phi\rangle\eup\cap\langle\Psi\rangle\edown=\emptyset$. By the PIT, we can find a prime filter $w_{\Phi}\supseteq \langle \Phi\rangle\eup$ in $\pf\mathcal{L}$ such that $w_\Phi\cap \langle \Psi\rangle\edown=\emptyset$. It follows immediately that
\[
w_\Phi\models \Phi\text{ and }w_\Phi\not\models\Psi
\]
which is what we wanted to show.
\end{proof}

\subsection{Modularity.}
The coalgebraic setting allows us to combine completeness-via-cano-nicity results from simple logics to get results for more 
complicated logics. It can be shown that the coalgebraic J\'{o}nsson-Tarski theorem is modular in the following sense.

\begin{theorem}[Strong completeness transfers under fusion]\label{thm:modularComp}
Let $L_i:\RK\to\RK, T_i:\SD\to\SD,\delta^i: L_iG\to GT_i, i=,1,2$. For any $(L_1+L_2)-$algebra $(A,\alpha)$, if $\hat{\delta}^i_A$ has a right inverse $\zeta^i_A, i=1,2$, then $\eta_A: A\to GFA$ lifts to an $L_1+L_2$-algebra morphism.
\end{theorem}
\begin{proof}
We show that the following diagram commutes: 
\[
\xymatrix
{
L_1 A+L_2 A\ar[r]^-{L_1\eta_A+L_2\eta_A}\ar[dddd]_{\alpha}\ar[dddr]_{\eta_{L_1A+L_2A}} & L_1 GFA + L_2 GFA\ar[d]^{\delta^1_A+\delta^2_A}\\
& GT_1 FA + GT_2 FA\ar[d]^{G\pi_1+G\pi_2}\\
&  G(T_1 FA\times T_2 FA)\ar[d]^{G(\zeta^1_A\times \zeta^2_A)}\\
& G(FL_1 A\times FL_2 A)\simeq GF(L_1 A+L_2 A)\ar[d]^{GF\alpha}\\
A\ar[r]_{\eta_A} & GFA
}
\] 
$F$ being left adjoint preserves colimits, and thus turns coproduct in $\RK$ into products in $\SD$. The bottom left-hand corner trapezium thus commutes by naturality of $\eta$. So we must show the commutativity of to top-right-hand corner triangle. For this we first show that
\[
(G\pi_1+G\pi_2)\circ (\delta^1+\delta^2)\circ (L_1\eta_A+L_2\eta_A)=G(\hat{\delta}^1_A\times \hat{\delta}^2_A)\circ\eta_{L_1 A+L_2 A}
\]
This is easily seen from the following diagram, which unravels the definition of adjoint transposes and uses the fact that $F$ preserves colimits: 
\[
\xymatrix@C=4ex
{
L_1A+L_2A\ar[d]_{L_1\eta_A\hspace{-2pt}+\hspace{-2pt}L_2\eta_A}\ar[r] & GF(L_1A\hspace{-2pt}+\hspace{-2pt}L_2A)\hspace{-2pt}\simeq\hspace{-2pt} G(FL_1A\hspace{-2pt}\times\hspace{-2pt} FL_2A)\ar[d]^{G(FL_1\eta_A\times FL_2\eta_A)}_{GF(L_1\eta_A+L_2\eta_A)\simeq}\\
L_1 GFA\hspace{-2pt}+\hspace{-2pt}L_2GFA\ar[d]_{(\delta^1)_{FA}+(\delta^2)_{FA}}\ar[r] & GF(L_1 GFA\hspace{-2pt}+\hspace{-2pt}L_2GFA)\hspace{-2pt}\simeq\hspace{-2pt} G(FL_1GFA\hspace{-2pt}\times\hspace{-2pt} FL_2GFA)\ar[d]^{G(F(\delta^1)_{FA}\times G(F(\delta^2)_{FA})}_{GF((\delta^1)_{FA}+(\delta^2)_{FA})\simeq}\\
GT_1FA\hspace{-2pt}+\hspace{-2pt}GT_2FA\ar[d]_{G\pi_1+G\pi_2}\ar[r] & GF(GT_1FA\hspace{-2pt}+\hspace{-2pt}GT_2FA)\hspace{-2pt}\simeq\hspace{-2pt} G(FGT_1FA\hspace{-2pt}\times\hspace{-2pt} FGT_2FA)\ar[dl]^{\hspace{8ex}G(\epsilon_{T_1FA}\circ \pi_1 \times\epsilon_{T_2FA}\circ \pi_2)}\\
G(T_1FA\hspace{-2pt}\times\hspace{-2pt} T_2FA)
}
\]
All the horizontal arrows are simply given by the unit $\eta: \Id\to GF$ (we have omitted the labels to keep the diagram readable), and thus the two top rectangles commute by naturality. Finally, we are left to deal with the bottom triangle which can be seen to commutes from the following commutative diagram:
\[
\xymatrix@C=7pt
{
GT_1 FA\ar[d]^{\eta_{GT_1FA}} \ar[r]^-{i_1}
& G(T_1FA\times T_2FA)\simeq GT_1FA+GT_2FA\ar[d]_{\eta_{GT_1FA+GT_2FA}}
& GT_2FA\ar[d]_{\eta_{GT_2FA}}\ar[l]_-{i_2}\\
GFGT_1FA \ar[r]^-{GF i_1}_-{\simeq G\pi_1}\ar[d]^{G\epsilon_{T_1FA}}
& G(FGT_1FA\hspace{-2pt}\times\hspace{-2pt} FGT_2FA)\hspace{-2pt}\simeq\hspace{-2pt} GF(GT_1FA\hspace{-2pt}+\hspace{-2pt}GT_2FA) \ar[d]^{G(\epsilon_{T_1FA}\circ \pi_1\times \epsilon_{T_2FA}\circ \pi_2)}
& GFGT_1FA\ar[l]_-{GF i_2}^-{\simeq G\pi_2}\ar[d]_{G\epsilon_{T_2FA}}\\
GT_1 FA \ar[r]_{G\pi_1}
& G(T_1FA\times T_2FA) 
& GT_2FA\ar[l]^{G\pi_2}
}
\]
The top squares commute by naturality of $\eta$, the bottom squares commute by naturality of $\epsilon$ and the two squares can be joined by the fact that $F$ turns coproducts into products. Note also that $G\epsilon_{T_1FA}\circ \eta_{GT_1 FA}=\Id_{GT_1FA}$ by the fact that $F\dashv G$, and the desired result follows from the unicity of the coproduct map $G\pi_1+G\pi_2$. It is now easy to see that
\begin{align*}
&G(\zeta^1_A\times \zeta^2_A)\circ (G\pi_1+G\pi_2)\circ (\delta^1+\delta^2)\circ (L_1\eta_A+L_2\eta_A)\\
=&G(\zeta^1_A\times \zeta^2_A)\circ G(\hat{\delta}^1_A\times \hat{\delta}^2_A)\circ\eta_{L_1 A+L_2 A}\\
=&G((\hat{\delta}^1_A\times \hat{\delta}^2_A)\circ (\zeta^1_A\times \zeta^2_A))\circ \eta_{L_1 A+L_2 A}\\
=&\eta_{L_1 A+L_2 A}, 
\end{align*}
by the assumption that $\zeta^1_A$ and $\zeta^2_A$ are right inverses.
\end{proof}
Note that we can extract a \emph{model} of the right type from the proof above, namely
\[
\zeta^1_A\times \zeta^2_A\circ F\alpha: FA\to T_1FA\times T_2FA.
\]

\section{Application to distributive substructural logics.}
\subsection{Describing $\TRL$-coalgebras validating FC\ref{ax:FrameCond1}-\ref{ax:FrameCond6}} \label{subsec:Describing}

The axioms FC\ref{ax:FrameCond1}-\ref{ax:FrameCond6} translate as two simple frame conditions on the relational (resource) 
models interpreting the logic defined by $\LRL$ and these axioms; one dealing with the unit $I$ of the language, and the other 
with the residuation of $\lRes$ and $\rRes$ with respect to $\fus$. The class of $\TRL$-coalgebras satisfying the second frame 
condition is described in Theorem \ref{thm:framecond}. It is intuitive and indeed corresponds to the usual relational semantics of 
distributive substructural logics (see, e.g.,  \cite{2002:restallintroduction}) or separation logic/BI. However, proving it `from first 
principles' as we do here is much more intricate than might be expected and, indeed, much more so than is clear in \cite{2015:self}. 

Let $\gamma:W\to \TRL W$ be a $\TRL$-coalgebra validating the axioms FC\ref{ax:FrameCond1}-\ref{ax:FrameCond6}. Axioms FC\ref{ax:FrameCond1} means that every world $w\in W$ must have amongst its successors pairs $(w,x)$ and $(y,w)$ such that $x,y$ are `unit states', viz. $x,y\models I$, moreover these are the only successors of $w$ containing a unit state.  This condition can be found in, for example, \cite{2007:calcagnoPOPL}. The other axioms are simply 
designed to capture the residuation condition in such a way that canonicity can be used; so a model in which FC\ref{ax:FrameCond2}-\ref{ax:FrameCond6} 
are valid is simply a model in which the residuation conditions hold, viz. $a\fus b\leq c$ iff $b\leq a\lRes c$ iff $a\leq c\rRes b$. 

To see what this means for $\TRL$-coalgebras we need the following lemma. For any $\TRL$-coalgebra $\gamma: W\to \mathbbm{2}\times \Pow(W\times W)\times \Pow(W\op\times W)\times \Pow(W\times W\op)$, let $\gamma_I,\gamma_\fus,\gamma_{\ulRes}$ and $\gamma_{\urRes}$ define the four components of the structure map. We define
\begin{align*}
\gamma_\fus^{(\downarrow\times \downarrow)}(w)&=\{(x,y)\mid \exists (x',y')\in \gamma_{\fus}(w), x\leq x', y\leq y'\}\\
\gamma_{\ulRes}^{(\downarrow\times \uparrow)}(w)&=\{(x,y)\mid \exists (x',y')\in \gamma_{\ulRes}(w), x\leq x', y'\leq y\}\\
\gamma_{\urRes}^{(\uparrow\times \downarrow)}(w)&=\{(x,y)\mid \exists (x',y')\in \gamma_{\urRes}(w), x'\leq x, y\leq y'\}
\end{align*}

\begin{lemma}\label{lem:updownClosure}
Let $\gamma: W\to \TRL W$ be a coalgebra and let $\hat{\gamma}: W\to \TRL W$ be the $\TRL$-coalgebra defined by $\gamma_I,\gamma_\fus^{(\downarrow\times \downarrow)}, \gamma_{\ulRes}^{(\downarrow\times \uparrow)}, \gamma_{\urRes}^{(\uparrow\times \downarrow)}$ defined as above, then for any valuation $v:\Free W\to \Upsets W$ and any $w\in W$
\[
(w,\gamma, v)\models a\text{ iff }(w,\hat{\gamma}, v)\models a
\]
\end{lemma}
\begin{proof}
This is an easy consequence of the the definition of $\TRL$ and of the fact that the denotation of any formula is an upset.
\end{proof}

We can now formulate the residuation condition.

\begin{lemma}\label{lem:residuationLemma}
Equations FC\ref{ax:FrameCond2}-\ref{ax:FrameCond6} are valid in $\gamma: W\to \TRL W$ iff
\begin{align*}
(y,z)\in \gamma_\fus^{(\downarrow\times \downarrow)}(x) \text { iff }(y,x)\in \gamma_{\ulRes}^{(\downarrow\times \uparrow)}(z) \text{ iff } (x,z)\in \gamma_{\urRes}^{(\uparrow\times \downarrow)}(y)
\end{align*}
\end{lemma}
\begin{proof}
The `if' direction follows easily by unravelling the definition of the semantics. So let us turn to the `only if' direction. Assume that $a\fus b\leq c$ iff $b\leq a\lRes c$ and let $\gamma: W\to \TRL W$. We will show that $(y,z)\in \gamma_\fus^{(\downarrow\times \downarrow)}(x) \text { iff }(y,x)\in \gamma_{\ulRes}^{(\downarrow\times \uparrow)}(z)$, the case of $(x,z)\in \gamma_{\urRes}^{(\uparrow\times \downarrow)}(y)$ is treated identically.

We start by showing that if $(y,x)\in \gamma_{\ulRes}^{(\downarrow\times \uparrow)}(z)$, then  $(y,z)\in \gamma_\fus^{(\downarrow\times \downarrow)}(x)$. Assume $(y,x)\in \gamma_{\ulRes}^{(\downarrow\times \uparrow)}(z)$ --- that is, that there exist $(y',x')\in\gamma_{\ulRes}(z)$ such that $y\leq y'$ and $x'\leq x$. Consider a valuation such that $\lsem a\rsem =\uparrow y$ and $\lsem b \rsem=\uparrow z$. We're assuming that $a\fus b\leq c$ iff $b\leq a\lRes c$, which means in particular that $b\leq a\lRes(a\fus b)$. Since $z\models b$, it must therefore also be the case that $z\models a\lRes(a\fus b)$. Since $\lsem a\rsem=y$ and $y\leq y'$ we have $y'\models a$ and it must therefore also be the case that $x'\models a\fus b$, and hence $x\models a \fus b$. It follows that there exist $(y'',z'')\in\gamma_\fus(x)$ with $y''\geq y$ and $z''\geq z$; that is, $(y,z)\in \gamma_\fus^{(\downarrow\times \downarrow)}(x)$.

For the converse, we show that if $(y,x)\notin \gamma_{\ulRes}^{(\downarrow\times \uparrow)}(z)$ then $(y,z)\notin \gamma_{\fus}^{(\downarrow\times \downarrow)}(x)$. Let $x,y,z\in W$ with $(y,x)\notin \gamma_{\ulRes}^{(\downarrow\times \uparrow)}(z)$ and consider a valuation such that $\lsem a\rsem=\uparrow y$ and $\lsem c\rsem=(\downarrow x)^c$. It follows that $z\models a\lRes c$. Indeed, assume the opposite; that is, that there exists $(y',x')\in \gamma_{\ulRes}(z)$ such that $y'\models a$ --- that is, $y\leq y'$ --- and $x'\not\models c$ --- that is, $x'\notin (\downarrow x)^c$; that is, $x'\leq x$. This means exactly that $(y,x)\in \gamma_{\ulRes}^{(\downarrow\times \uparrow)}(z)$, a contradiction. Thus $z\models a\lRes c$. Now assume that $(y,z)\in\gamma_\fus^{(\downarrow\times \downarrow)}(x)$; that is, there exist $(y'',z'')\in \gamma_\fus(x)$ such that $y''\geq y, z''\geq z$. Since $y\models a$ and $z\models a\lRes c$ we have $y''\models a$ and $z''\models a\lRes c$ and therefore we have $x\models a\fus (a\lRes c)$. Since we're assuming that $a\fus b\leq c$ iff $b\leq a\lRes c$, we have in particular that $a\fus (a\lRes c)\leq c$, and thus $x\models c$. But this is impossible since $\lsem c\rsem=(\downarrow x)^c$. Thus we cannot have $(y,z)\in\gamma_\fus(x)^{(\downarrow\times \downarrow)}$; that is, $(y,z)\notin\gamma_\fus(x)^{(\downarrow\times \downarrow)}$.
\end{proof}

The entire information required to encode a $\TRL$-coalgebra validating FC\ref{ax:FrameCond2}-\ref{ax:FrameCond6} is therefore entirely contained in $\gamma_\fus^{(\downarrow\times \downarrow)}$ (or $\gamma_{\ulRes}^{(\downarrow\times \uparrow)}$ or $\gamma_{\urRes}^{(\uparrow\times \downarrow)}$). We will now show that we can in fact simply consider $\gamma_\fus$ (or $\gamma_{\ulRes}$ or $\gamma_{\urRes}$). Note first that Lemma \ref{lem:residuationLemma} enforces a strict constraint on $\gamma_\fus^{(\downarrow\times\downarrow)}$: since $\gamma_{\ulRes}^{(\downarrow\times\uparrow)},\gamma_{\urRes}^{(\uparrow\times \downarrow)}$ are monotone it follows that if $z\leq z'$  we must have
\begin{align}
\{(y,x) \mid (y,z)\in \gamma_\fus^{(\downarrow\times\downarrow)}(x)\} \hspace{1ex}\leq_{W\op\times W}\hspace{1ex} 
\{(y',x') \mid (y',z')\in \gamma_\fus^{(\downarrow\times\downarrow)}(x')\}
\label{eq:EM1}
\end{align}
and similarly, if $y\leq y'$
\begin{align}
\{(x,z) \mid (y,z)\in \gamma_\fus^{(\downarrow\times\downarrow)}(x)\} \hspace{1ex}\leq_{W\times W\op}\hspace{1ex} 
\{(x',z') \mid (y',z')\in \gamma_\fus^{(\downarrow\times\downarrow)}(x')\}\label{eq:EM2}
\end{align}
for the Egli-Milner order on $W\op\times W$. The inequations (\ref{eq:EM1}) and (\ref{eq:EM2}) are quite strong and must be satisfied by any map $\gamma_\fus^{(\downarrow\times\downarrow)}$ capable of reconstructing the entire $\TRL$-coalgebra. More generally, we say that a a monotone map $\gamma: W\to \Pow(W\times W)$ obeys  the \emph{Residuation Compatibility Condition} (RCC) if
\begin{align}
& \{(y,x)\mid (y,z)\in\gamma(x)\}\leq_{W\op\times W}\{(y',x')\mid (y',z')\in\gamma(x')\}\text{ and }\nonumber  \\
& \{(x,z)\mid (y,z)\in\gamma(x)\}\leq_{W\times W\op}\{(x',z')\mid (y',z')\in\gamma(x')\}\label{eq:RCC}\tag{RCC}
\end{align}
We now generalize the statement of Lemma \ref{lem:residuationLemma} to an arbitrary monotone map $\gamma: W\to \Pow(W\times W)$ by defining
\begin{align}
\overline{\gamma}: W\to \Pow(W\op\times W), z\mapsto\{(y,x)\mid (y,z)\in\gamma(x)\}\label{eq:overline}\\
\underline{\gamma}: W\to \Pow(W\times W\op), y\mapsto\{(x,z)\mid (y,z)\in\gamma(x)\}\label{eq:underline}
\end{align}
\begin{lemma}
If $\gamma:W\to \Pow(W\times W)$ satisfies (\ref{eq:RCC}), then $\overline{\gamma}$ and $\underline{\gamma}$ are monotone maps.
\end{lemma}
\begin{proof}
Immediate from the definitions.
\end{proof}
With this notation in place we can now state the main result of this Section, which relies on showing that we can impose (\ref{eq:RCC}) on $\gamma_\fus$ rather than $\gamma_\fus^{(\downarrow\times \downarrow)}$.
\begin{theorem}\label{thm:framecond}
The logic defined by $\LRL$ and the axioms FC\ref{ax:FrameCond2}-\ref{ax:FrameCond6} is strongly complete with respect the class of $\TRL$-coalgebras of the shape
\[
\gamma_I\times\gamma_\fus\times\overline{\gamma}_\fus\times\underline{\gamma}_\fus: W\to \mathbbm{2}\times \Pow(W\times W)\times \Pow(W\op\times W)\times \Pow(W\times W\op)
\] 
such that $\gamma_\fus$ satisfies (\ref{eq:RCC}).
\end{theorem}
\begin{proof}
Let $\gamma=\gamma_I\times \gamma_\fus\times \gamma_{\ulRes}\times \gamma_{\urRes}: W\to \TRL W$ be a coalgebra in which the frame conditions FC\ref{ax:FrameCond2}-\ref{ax:FrameCond6} are valid. From Lemma \ref{lem:updownClosure} it follows that the frame conditions are valid in the coalgebra 
\[
\gamma_I\times \gamma_\fus^{(\downarrow\times \downarrow)}\times \gamma_{\ulRes}^{(\downarrow\times \uparrow)}\times \gamma_{\urRes}^{(\uparrow\times \downarrow)}: W\to \TRL W
\]
and from Lemma \ref{lem:residuationLemma} this coalgebra is in fact of the shape
\[
\gamma_I\times \gamma_\fus^{(\downarrow\times \downarrow)}\times \overline{\gamma_\fus^{(\downarrow\times \downarrow)}}\times\underline{\gamma_\fus^{(\downarrow\times \downarrow)}}: W\to \TRL W
\]
If we can show that $ \overline{\gamma_\fus^{(\downarrow\times \downarrow)}}= \overline{\gamma_\fus}^{(\downarrow\times \downarrow)}$  and $ \underline{\gamma_\fus^{(\downarrow\times \downarrow)}}= \underline{\gamma_\fus}^{(\downarrow\times \downarrow)}$, then our claim will follow from Lemma \ref{lem:updownClosure}. We show the first equality, the second one being similar. 
\begin{align*}
&(y,x)\in \overline{\gamma_\fus^{(\downarrow\times \downarrow)}}(z)\\
\Leftrightarrow & (y,z)\in \gamma_\fus^{(\downarrow\times \downarrow)}(x) &\text{Def. of }\overline{\gamma_\fus^{(\downarrow\times \downarrow)}}\\
\Leftrightarrow & \exists y',z'\text{ s.th. } y\leq y', z\leq z' \text{ and } (y',z')\in \gamma_\fus(x)&\text{Def. of }\gamma^{(\downarrow\times\downarrow)}_\fus\\
\Leftrightarrow  & \exists y',z' \text{ s.th. }y\leq y', z\leq z' \text{ and } (y',x)\in\overline{\gamma_\fus}(z')&\text{Def. of }\overline{\gamma}_\fus \\
\Rightarrow & \exists x',y' \text{ s.th. } x'\leq x, y\leq y' \text{ and }(y',x')\in\overline{\gamma_\fus}(z)&(\ref{eq:RCC})\text{ and }z\leq z'\\
\Leftrightarrow & (y,x)\overline{\gamma_\fus}^{(\downarrow\times \uparrow)}(z)&\text{ Def. of }\overline{\gamma_\fus}^{(\downarrow\times \uparrow)}
\end{align*}
Conversely, we have
\begin{align*}
& (y,x)\overline{\gamma_\fus}^{(\downarrow\times \uparrow)}(z)\\
\Leftrightarrow & \exists x',y'\text{ s.th. }y\leq y', x'\leq x\text{ and }(y',x')\in\overline{\gamma_\fus}(z)&\text{Def. of }\overline{\gamma}_\fus^{(\downarrow\times \uparrow)}\\
\Leftrightarrow  & \exists x',y'\text{ s.th. }y\leq y', x'\leq x\text{ and }(y',z)\in \gamma_\fus(x')&\text{Def. of }\overline{\gamma}_\fus\\
\Rightarrow & \exists y',z'\text{ s.th. }y\leq y', z\leq z'\text{ and }(y',z')\in\gamma_\fus(x) & \gamma_\fus \text{ monotone and }x'\leq x\\
\Leftrightarrow & (y,z)\in\gamma_{\fus}^{(\downarrow\times\downarrow)}(x)&\text{Def. of }\gamma^{(\downarrow\times\downarrow)}_\fus\\
\Leftrightarrow & (y,x)\in \overline{\gamma_{\fus}^{(\downarrow\times\downarrow)}}(z)&\text{Def. of }\overline{\gamma_{\fus}^{(\downarrow\times\downarrow)}}
\end{align*}
\end{proof}

\begin{example}\label{example:heaps}
Heaps, which for the sake of brevity and convenience we shall define as partial maps on $\mathbb{N}_+$ with finite domain, form a model satisfying Theorem \ref{thm:framecond}. To show this, we first define the set of heaps
\[
\mathcal{H}=\{f:\mathbb{N}_+\rightharpoonup\mathbb{N}_+\mid \dom(f)\in\mathcal{P}_{\omega}(\mathbb{N}_+)\}
\]
It forms a poset under the order
\[
f\leq g\text{ if }f=g\restrict\dom(f)
\]
We now define a $\TRL$-coalgebra structure on $\mathcal{H}$ as follows:
\begin{alignat*}{3}
 & \gamma_I : \mathcal{H}\to\mathbbm{2},  && \hspace{2ex}f \mapsto && \begin{cases}0&\text{ if }\dom(f)=\emptyset\\1&\text{ else}\end{cases}  \\
& \gamma_{\fus}: \mathcal{H}\to \Pow(\mathcal{H}\times\mathcal{H}), &&\hspace{2ex} f  \mapsto && \{(g,h)\mid \dom (g)\cap\dom(h)=\emptyset, \\
&  && && f\restrict\dom(g)=g, f\restrict\dom(h)=h\} 
\end{alignat*}
Clearly $\gamma_I$ is trivially monotone, and validates the axioms FC\ref{ax:FrameCond1}. To see that $\gamma_{\fus}$ is well-typed, note first that $\gamma_{\fus}(f)$ is a down-set, and therefore also convex. Moreover, it is not hard to see that if $f'$ extends $f$; that is, $f\leq f'$ then $\gamma_{\fus}(f)\subseteq \gamma_{\fus}(f')$. The first half of the Egli-Milner definition is therefore trivially satisfied. For the second half, if $(g',h')\in \gamma_{\fus}(f')$, then $(g'\restrict \dom(f), h'\restrict\dom(f)\in\gamma_{\fus}(f)$ provides the witness we need. It follows that $\gamma_{\fus}$ is also monotone and thus well-typed. Finally, we need to check that it satisfies (\ref{eq:RCC}); that is, if $h\leq h'$ 
\[
\{(g,f)\mid (g,h)\in\gamma_{\fus}(f)\}\leq_{\mathcal{H}\op\times\mathcal{H}}\{(g',f')\mid (g',h')\in\gamma_{\fus}(f')\}
\]
Starting from $(g,f)$ in the first set, we define $f'$ by $f'=f$ on $\dom(f)$ and $f'=h'$ on $\dom(h')\cap(\dom(f))^c$ and get an element $(g,f')$ such that $(g,h')\in \gamma_{\fus}(f'), g\leq g$ and $f\leq f'$, which shows that the first direction of the Egli-Milner order holds. For the second, start with $(g',f')$ in the second set and define $f$ by $f=h$ on $\dom(h)$ and $f=g'$ on $\dom(g')$. It is easy to see that $(g',h)\in \gamma_{\fus}(f),g'\leq g'$ and $f\leq f'$, providing a witness for the second direction of the Egli-Milner order. It follows that $\gamma_{\fus}$ satisfies (\ref{eq:RCC}), and heaps therefore provide a class of models for the distributive full Lambek calculus.
\end{example}
\subsection{Additional frame conditions.} 

Structural rules can be added to the full distributive Lambek calculus to form new logics. These rules are the exchange rule (e), the contraction rule (c), the left weakening rule (lw), and the right weakening rule (rw). Relevance logic for example consists in adding (c) and (e) to the distributive Lambek calculus, adding only (c) defines the positive MALL$^+$ fragment of linear logic (\cite{2002:restallintroduction}), whilst the combination of (lw), (rw) and (e) defines affine logic. These structural rules correspond to (in)equations in the theory of residuated lattices (see \cite{2002:restallintroduction,2003:onoresiduated,2007:galatosresiduated}); that is, in the language of $\LRL$-algebras. Let us go through them in order.

\medskip 

\textbf{Exchange.} The exchange rule corresponds to the axiom (e) given by $a\fus b=b\fus a$. It is easy to see from the results of Section \ref{sec:DLE} that it is canonical. The class of $\TRL$-coalgebras in which this axiom is valid is characterised by 
\[
(y,z)\in\gamma_\fus(x)\Rightarrow (z,y)\in\gamma_\fus^{(\downarrow\times\downarrow)}(x)
\] 
in other words, the ternary relation defined by $\gamma_\fus^{(\downarrow\times\downarrow)}$ is closed under permuting the successor states. In the case of the boolean Lambek calculus, that is to say in the classical case, we clearly have $(y,z)\in\gamma_\fus(x)\text{ iff }(z,y)\in\gamma_\fus(x)$.

\medskip 

\textbf{Contraction}. The contraction rule corresponds to the axiom (c) given by $a\leq a\fus a$; that is, increasing idempotency (see \cite{2003:onoresiduated,2007:galatosresiduated}). Once again, the canonicity of this axiom is almost immediate from Section \ref{sec:DLE}. The class of $\TRL$-coalgebras in which (c) is valid is characterized by 
\[
\forall x\in W \exists (y,z)\in\gamma(x)\text{ s.th. }x\leq y,z
\]
and in the classical case, this means that $(x,x)\in\gamma_{\fus}(x)$.

\textbf{Left weakening}. The weakening rule corresponds to the axiom $a\leq I$, viz. every state is a unit state. Coalgebraically, this means that the component $\gamma_I$ of a $\TRL$-coalgebra is the constant map $0\in\mathbbm{2}$. In the classical case, this amounts to saying that $I=\top$.

For the right weakening rule we need to introduce a new unit, which we will denote $J$. In fact we use this as an opportunity to introduce a whole new signature, dual to the signature defining $\LRL$. We define
\[
\LRL\dual: \RK\to\RK, \begin{cases}
\LRL\dual A =\Free\{J, a\fusd b, a\lResd b, a\rResd b\mid a,b\in \Forg A\}/\equiv\\
\LRL\dual f: \LRL A\to \LRL B, [a]_{\equiv}\mapsto [f(a)]_{\equiv} \, , 
\end{cases}
\]
where $\equiv$ is the fully invariant equivalence relation in $\RK$ generated by following the Distribution Laws for non-empty finite subsets $X$ of $A$:
\begin{multicols}{2}
\begin{enumerate}[DL$\dual$1.]
\item$\bigwedge X\fusd a=\bigwedge [X\fusd a]\label{ax:distribLaw1d} $ 
\item$a\fusd \bigwedge X=\bigwedge[a\fusd X]$
\label{ax:distribLaw2d} 
\item$a\lResd\bigvee X=\bigvee[a\lResd X]$
\label{ax:distribLaw3d} 
\item$\bigwedge X\lResd a=\bigvee[X\lResd a]$
\label{ax:distribLaw4d} 
\item$\bigvee X\rResd a=\bigvee [X\rResd a]$
\label{ax:distribLaw5d} 
\item$a\rResd \bigwedge X=\bigvee[a\rResd X]$.
\label{ax:distribLaw6d}
\end{enumerate}
\end{multicols}
Note that the equations DL$\dual$\ref{ax:distribLaw1d}-\ref{ax:distribLaw6d} are dual to the equations DL\ref{ax:distribLaw1}-\ref{ax:distribLaw6}. In particular, $\fusd$ is a binary $\square$, whilst $\fus$ is a binary $\dia$. The language $\lang(\LRL\dual)$ will also be interpreted in $\TRL$ coalgebras, via the semantic transformation $\semRLd:\LRL\dual G\to G\TRL$ defined at every poset $W$ via its action on the generators:
\begin{align*}
\semRLd_W(J)&=\{t\in \TRL W \mid \pi_1(t)=1\in\two\}\\
\semRLd_W(u\fusd v)&=\{t\in \TRL W\mid \forall (x,y)\in \pi_2(t), x\in u\text{ or } y\in v\}\\
\semRLd_W(u\lResd w)&=\{t\in \TRL W\mid \exists (x,y)\in \pi_3(t), x\notin u\text{ and } y\in w\}\\
\semRLd_W(w\rResd v)&=\{t\in \TRL W\mid\exists (x,y)\in \pi_4(t), x\notin v\text{ and } y\in w\}
\end{align*}
A proof very similar to that of Proposition \ref{prop:semPreservationProp} shows that $\semRLd$ is well-defined, and using exactly the same technique as in Theorem \ref{thm:rightInvDelta}, it can be shown that the adjoint transpose of $\semRLd$ has right inverses at every distributive lattice. It follows that $\lang(\LRL\dual)$ is strongly complete with respect to the class of all $\TRL$-coalgebras. However, the intended interpretation of $\lResd$ and $\rResd$ is once again that they should be the left and right residuals of $\fusd$. We therefore introduce the following axioms.
\begin{multicols}{2}
\begin{enumerate}[FC$\dual$1.]
\item $a\fusd J=a$, $J\fusd a=a$, 
\label{ax:FrameCond1d} 
\item $I\leq a\lResd a$, $I\leq a\rResd a$,
\label{ax:FrameCond2d} 
\item $a\fusd(b\lResd c)\leq (a\fusd b)\lResd c$,\label{ax:FrameCond3d} 
\item $(c\rResd b)\fusd a\leq c\rResd(a\fusd b)$,\label{ax:FrameCond4d} 
\item $(a\rResd b)\fusd b\leq a$, and 
\label{ax:FrameCond5d} 
\item $b\fusd(b\lResd a)\leq a$.
\label{ax:FrameCond6d}
\end{enumerate}
\end{multicols}

These axioms capture identities and residuation, and are therefore of exactly the same shape as axioms FC\ref{ax:FrameCond1}-\ref{ax:FrameCond6}. A dual construction to that of  Section \ref{subsec:Describing}, shows that equations FC$\dual$\ref{ax:FrameCond1d}-\ref{ax:FrameCond6d} are valid in a $\TRL$-coalgebra $\gamma: W\to\TRL$ iff it satisfies
\begin{align*}
(y,z)\in \gamma_\fusd^{(\uparrow\times \uparrow)}(x) \text { iff }(y,x)\in \gamma_{\ulResd}^{(\uparrow\times \downarrow)}(z) \text{ iff } (x,z)\in \gamma_{\urResd}^{(\downarrow\times \uparrow)}(y)
\end{align*}
where $\overline{\gamma}_\fusd$ and $\underline{\gamma}_\fusd$ are defined by Eqs. (\ref{eq:overline}), (\ref{eq:underline}). The fact that $\gamma_{\ulResd}^{(\uparrow\times \downarrow)}$ and $\gamma_{\urResd}^{(\downarrow\times \uparrow)}$ should be monotone means that $\gamma_{\fusd}^{(\uparrow\times \uparrow)}$ must satisfy (\ref{eq:RCC}). Finally, dualising Theorem \ref{thm:framecond} we get 
\begin{theorem}
$\lang(\LRL\dual)$ quotiented by the axioms FC$\dual$\ref{ax:FrameCond2d}-\ref{ax:FrameCond6d} is strongly complete with respect to the class of $\TRL$-coalgebras of the shape
\[
\gamma_J\times\gamma_\fusd\times\overline{\gamma}_\fusd\times\underline{\gamma}_\fusd: W\to \mathbbm{2}\times \Pow(W\times W)\times \Pow(W\op\times W)\times \Pow(W\times W\op)
\] 
such that $\gamma_{\fusd}$ satisfies (\ref{eq:RCC}).
\end{theorem}

Combining this theorem with our earlier result on the modularity of coalgebraic completeness-via-canonicity (Theorem \ref{thm:modularComp}) we get that the logic defined by $\lang(\LRL+\LRL\dual)$ and the axioms FC\ref{ax:FrameCond2}-\ref{ax:FrameCond6} and FC$\dual$\ref{ax:FrameCond2d}-\ref{ax:FrameCond6d} is strongly complete with respect to $\TRL\times \TRL$-coalgebras of the shape:
\begin{align}
(\gamma_I & \times\gamma_\fus\times\overline{\gamma}_\fus\times\underline{\gamma}_\fus)\times(\gamma_J\times\gamma_\fusd\times\overline{\gamma}_\fusd\times\underline{\gamma}_\fusd):\nonumber\\ W\to (\mathbbm{2}\times & \Pow(W\times W)\times \Pow(W\op\times W)\times \Pow(W\times W\op))^2\label{eq:MALL+}
\end{align}
where both $\gamma_{\fus}$ and $\gamma_{\fusd}$ satisfy (\ref{eq:RCC}). We can now return to the structural rules which we described at the beginning of this section and account for right weakening.

\medskip

\textbf{Right weakening.} Working in the signature of the fusion $\lang(\LRL\oplus\LRL\dual)$, right weakening corresponds to the axioms $J\leq a$ which is clearly canonical. Coalgebraically this axiom corresponds to saying that if $\gamma_J(w)=1$ in a model, then $w\models a$ for any formula $a$ and any valuation. This is clearly not possible: let $p\in V$ be a propositional variable and consider a valuation such that $\lsem p\rsem=(\downarrow w)^c$ (a valid upset), clearly $w\not\models p$. It follows that $\gamma_J$ must be the constant monotone map $0\in\mathbbm{2}$. In particular if left weakening is also allowed, then nothing distinguishes $\gamma_I$ and $\gamma_J$ and $J$ holds precisely when $I$ does not; that is, never. In the classical case we clearly have $J=\bot$.

The logic defined by $\LRL+\LRL\dual$ and the axioms FC\ref{ax:FrameCond1}-\ref{ax:FrameCond6}, FC$\dual$\ref{ax:FrameCond1d}-\ref{ax:FrameCond6d} and the exchange axioms $a\fus b=b\fus a$ and $a\fusd b=b\fusd a$ is the positive fragment of the Multiplicative-Additive Linear Logic (MALL$^+$ in \cite{2002:restallintroduction}). Many additional features could be added to the quotient of $\lang(\LRL+\LRL\dual)$ under FC\ref{ax:FrameCond1}-\ref{ax:FrameCond6} and FC$\dual$\ref{ax:FrameCond1d}-\ref{ax:FrameCond6d}, most notably one could define the `negation' operations $\sim a\triangleq a\lRes J$ and $\neg a\triangleq J\rRes a$ and use them to connect the behaviour of the two halves of the signature. We refer the reader to \cite{2002:restallintroduction},\cite{2003:onoresiduated} and \cite{2007:galatosresiduated} for such considerations. To conclude we return to our heap model and show that it is a model for the logic we have just defined.

\begin{example}[\cite{2014:Brotherston-Villard}]
Recall from Example \ref{example:heaps} that the poset of heaps $\mathcal{H}$ can be equipped with a map $\gamma_{\fus}: \mathcal{H}\to\Pow(\mathcal{H}\op\times \mathcal{H})$ which satifies (\ref{eq:RCC}), and can thus be used to reconstruct an entire $\TRL$-coalgebra structure (modulo a monotone map $\gamma_I: \mathcal{H}\to \mathbbm{2}$). We will now define a second such map which will interpret the dual signature given by $\LRL\dual$ in the way suggested by \cite{2014:Brotherston-Villard}. We choose an arbitrary upset of heaps $U\subseteq \mathcal{H}$ and define
\begin{alignat*}{3}
 & \gamma_J : \mathcal{H}\to\mathbbm{2},  & & f \mapsto & & \begin{cases}1&\text{ if }f\in U \\  
 	s0 & \text{ else}\end{cases}  \\
 & \gamma_{\fusd}: \mathcal{H}\to \Pow(\mathcal{H}\times\mathcal{H}), & & f  \mapsto & & 
 	\{(g,h)\mid \dom(f)=\dom(g)\cap\dom(h) \\
 & & & & & \quad g\restrict\dom(f)=h\restrict\dom(f)\} 
\end{alignat*}
The map $\gamma_J$ is well-typed by construction. To see that $\gamma_{\fusd}$ is also well-typed, note first that $\gamma_{\fusd}(f)$ is this time an upset, and therefore convex. It is not difficult to check in the same way as in Example \ref{example:heaps} that if $f\leq f'$ then $\gamma_{\fusd}(f)\leq \gamma_{\fusd}(f')$ for the Egli-Milner order; that is, $\gamma_{\fusd}$ is monotone and thus well-typed. Simple set-theoretic considerations of the same type as in Example \ref{example:heaps} also show that $\gamma_{\fusd}$ defined as above satisfies (\ref{eq:RCC}). It follows that the data of $\gamma_{\fus},\gamma_{\fusd}$ and $U$ endows the poset $\mathcal{H}$ of heaps with a $(\TRL)^2$-coalgebra structure allowing the interpretation of $\lang(\LRL+\LRL\dual)$-formulas. It is shown in \cite{2014:Brotherston-Villard} that several additional $\lang(\LRL+\LRL\dual)$-axioms may be envisaged. Most notably the axioms $a\leq a\fusd 0$ and $a\fusd 0\leq a$ (which combined give FC$\dual$\ref{ax:FrameCond1d}, an axiom we have chosen not to enforce as `standard'), the contraction axiom $a\fusd a\leq a$ (dual to axiom (c)) and the \emph{weak distribution axiom} $a\fus(b\fusd c)\leq (a\fus b)\fusd a$. These axioms are canonical, and Theorems \ref{thm:strCompMain} and \ref{thm:modularComp} therefore easily provides strong completeness for these axioms too, although as was is shown in \cite{2014:Brotherston-Villard}, no heap model can validate all three of these axioms simultaneously. 
\end{example}

\section{Conclusion and future work}

We have shown how distributive substructural logics can be formalized and given a semantics in the framework of 
coalgebraic logic, and highlighted the modularity of this approach. By choosing a syntax whose operators explicitly 
follow distribution rules, we can use the elegant topological theory of canonicity for DLs, and in particular the notion 
of smoothness and of topology matching, to build a set of canonical (in)equation capturing the 
distributive full Lambek calculus. The coalgebraic approach makes the connection between algebraic 
canonicity and canonical models explicit, categorical and generalizable.

The modularity provided by our approach is twofold. Firstly, we have a generic method for building canonical 
(in)equations by using the Principle of Matching Topologies. Getting completeness results with respect to simple 
Kripke models for variations of the distributive full Lambek calculus (e.g. relevant logic, MALL$^{+}$, etc...) becomes very 
straightforward. Secondly, adding more operators to the fundamental language simply amounts to taking a 
\emph{coproduct} of syntax constructors (e.g., $\LRL+\LML$ to define modal substructural logics) and interpreting it with a 
\emph{product} of model constructors (e.g., $\TRL\times\TML$). This is particularly suited to logics which 
build on BI such as the bi-intuitionistic boolean BI of \cite{2014:Brotherston-Villard}.

The operators $\fus,\lRes,\rRes$ all satisfy simple distribution laws, but our approach could also 
accommodate operators with more complicated distribution laws and non-relational semantics. For example, 
the theory presented in this work could perhaps be extended to cover a graded version of $\fus$, say $\fus_k$, whose 
interpretation would be `there are at least $k$ ways to separate a resource such that...' or `a resource can be split in two at a cost of $k$...', the semantics would 
be given by coalgebras of the type $\two\times\mathcal{B}(-\times -)$ where $\mathcal{B}$ is the `bag' or 
multiset functor. Similarly, a graded version $\to_k$ of the intuitionistic implication whose meaning would 
be `... implies ... apart from at most $k$ exceptions' and interpreted by $\mathcal{B}(-\times -)$-coalgebras could 
also be covered by our approach. Crucially, such operators do satisfy (more complicated) distribution laws 
which lead to generalizations of the results in Section \ref{sec:DLE}, and the possibility of building canonical 
(in)equations. The coalgebraic infrastructure would then allow the rest of the theory to stay essentially unchanged. We are currently investigating these possibilities.

\bibliographystyle{amsalpha} 
\bibliography{ResourcesBib}

\section*{Appendix}

\begin{proof}[Proof of Theorem \ref{thm:rightInvDelta}]
We prove the result for $\LRL$ and $\TRL$, the same technique can readily be applied to $\LML$ and $\TML$. We need to prove that $\hat{\delta}^{\mathrm{RL}}$ has a natural right inverse. By describing a prime filter of $LA$ in terms of the `generators' it contains we get the following characterization of $\hat{\delta}^{\mathrm{RL}}_A: T\pf A\to \pf LA$: 
\begin{align*}
I\in \hat{\delta}_A(U_1,U_2,U_3,x)&\text{ iff } x=0 \\
a\fus b\in \hat{\delta}_A(U_1,U_2,U_3,x)&\text{ iff }\exists (F_1,F_2)\in U_1, (a,b)\in (F_1,F_2)\\
a\lRes b\in \hat{\delta}_A(U_1,U_2,U_3,x)&\text{ iff }\forall (F_1,F_2)\in U_2,a\in F_1\Rightarrow b\in F_2\\
a\rRes b\in \hat{\delta}_A(U_1,U_2,U_3,x)&\text{ iff }\forall (F_1,F_2)\in U_3, b\in F_2\Rightarrow a\in F_1
\end{align*}

At every distributive lattice $A$, we now define $\gamma_A:\pf L A\to T\pf A$ by
\begin{align*} 
&F\mapsto\begin{cases}
&0\text{ if }I\in F, 1\text{ else}\\
&\{(F_1,F_2)\in (\pf A)^2\mid a\fus b \in F\text{ whenever }(a,b)\in (F_1,F_2)\}\\
&\{(G_1,G_2)\in (\pf A)^2\mid a\in G_1\Rightarrow b\in G_2\text{ whenever }a\lRes b \in F\}\\
&\{(H_1,H_2)\in (\pf A)^2\mid b\in H_2\Rightarrow a\in H_1\text{ whenever }a\rRes b \in F\} .
\end{cases}
\end{align*}
By unravelling the definition of $\hat{\delta}^{\mathrm{RL}}_A$, we get that for $\gamma_A$ to be a right inverse of $\hat{\delta}^{\mathrm{RL}}_A$ it must satisfy:
\begin{align}
(a\fus b)\in F & \Leftrightarrow  \exists (F_1,F_2)\in \pi_1(\gamma_A(F))\text{ s.th. }(a,b)\in(F_1,F_2)\label{eq:ExLemm1}\\
a\lRes b\in F & \Leftrightarrow  \forall G_1,G_2\in \pi_2(\gamma_A(F))\hspace{1em} a\in G_1\Rightarrow b\in G_2\label{eq:ExLemm2}\\
a\rRes b\in F & \Leftrightarrow  \forall H_1,H_2\in \pi_2(\gamma_A(F))\hspace{1em} b\in H_2\Rightarrow a\in H_1\label{eq:ExLemm3}
\end{align}
Note that the first component of $\gamma_A$ poses no difficulty since
\[
I \in \hat{\delta}_A(\gamma_A(F)) \Leftrightarrow \pi_1(\gamma_A(F))=0 \Leftrightarrow I\in F
\]
Note also that the right-to-left direction of (\ref{eq:ExLemm1}) and the left-to-right direction of (\ref{eq:ExLemm2},\ref{eq:ExLemm3}) follows immediately from the definitions. The hard part of the proof are the opposite directions.

\textbf{Left-to-right direction of (\ref{eq:ExLemm1}):} Assume that $a\fus b\in F$, we need to build $F_1,F_2$ such that: (1)$a\in F_1$, (2)$b\in F_2$ and (3)$(F_1,F_2)\in \pi_2(\gamma_A(F))$; that is,  $a'\in F_1$ and $b'\in F_2$ implies $a'\fus b'\in F$, or equivalently, $a'\fus b'\notin F$ implies $a'\notin F_1$ or $b'\notin F_2$. We will build $F_1,F_2$, using a proof which is similar to the proof of the prime ideal theorem for filter-ideal pairs. Let us denote by $\mathscr{P}(a,b)$ the set of pairs $((F_1,I_1),(F_2,I_2))$ such that
\begin{enumerate}
\item $\uparrow a\subseteq F_1$
\item $\uparrow b\subseteq F_2$
\item $I_1=\{c\mid \exists d\in F_2 \text{ s.th. }c\fus d\notin F\}$
\item $I_2=\{d\mid \exists c\in F_1 \text{ s.th. }c\fus d\notin F\}$
\item $F_1\subseteq\{c\mid \forall d\in F_2, c\fus d\in F\}$
\item $F_2\subseteq \{d\mid \forall c\in F_1, c\fus d\in F\}$
\end{enumerate}
We make the following observations about $\mathscr{P}(a,b)$
\begin{itemize}
\item it is non-empty: $((\uparrow a,\{c\mid c\fus b\notin F\}),(\uparrow b,\{d\mid a\fus d\notin F\}))\in \mathscr{P}(a,b)$
\item it forms a poset under component-wise set inclusion.
\item $I_1, I_2$ are ideals. It is clear that they are down-sets. Moreover, if $c,c'\in I_1$ then there exist $d,d'\in F_2$ s.th. $c\fus d, c'\fus d'\notin F$, and as a consequence $c\fus(d\wedge d'), c'\fus(d\wedge d')\notin F$ and $d\wedge d'\in F_2$ since $F_2$ is a filter. Since $F$ is prime and $\fus$ distributes over joins it follows that $(c\vee c')\fus (d\wedge d')\notin F$, and thus $d\wedge d'$ witnesses the fact that $c\vee c'\in I_1$. The proof is identical for $I_2$.
\item $F_1\cap I_1=F_2\cap I_2=\emptyset$ for each $((F_1,I_1), (F_2,I_2))\in\mathscr{P}(a,b)$. Indeed assume that there exist $f\in F_1,i\in I_1$ such that $f\leq i$, then we have $f\fus d\leq i\fus d$ for some $d\in F_2$ such that $i\fus d\notin F$. But by construction we must have $f\fus d\in F$ which contradicts $i\fus d\notin F$, since $F$ is a filter.
\end{itemize}

Let us now check that $\mathscr{P}(a,b)$ has upper bounds of chains. Assume that $((F_1^i,I_1^i),(F_2^i,I_2^i))\in \mathscr{P}(a,b), i\in \omega$ and define 
\[
F_i^\infty=\bigcup_j F_i^j, \hspace{2em}I_i^\infty=\bigcup_j I_i^j, \hspace{2em}i=1,2
\]
It is well-known and easy to check that the union of a chain of filter (resp. ideals) is a filter (resp. ideals). Let us now check that conditions 1.-6. are satisfied too. The first two conditions are trivially satisfied. For 3.-4., let $x\in I_1^\infty$, by definition there exists $i\in \omega$ s.th. $c\in I_1^i$ and thus there exist $d\in F_1^i$ s.th. $c\fus d\notin F$, but clearly $d\in F_1^\infty$ too, which shows that $I_1^\infty\subseteq \{c\mid d\in F_2^\infty \text{ s.th. }c\fus d\notin F\}$. The opposite inclusion works in exactly the
same way: let $c$ be s.th. there exists $d\in F_2^\infty$ s.th. $c\fus d\notin F$, then this $d$ can be traced back to a certain $F_2^i$ and thus $c\in I_1^i$. The proof for 5.-6. is very similar, let $c\in F_1^\infty$, then there exist $i\in\omega$ s.th. $c\in F_1^i$. Now let $d\in F_2^\infty$, then there exist $j\in \omega$ s.th. $d\in F_2^j$. By taking $k=\max(i,j)$ we get that $c\in F_1^k,d\in F_2^k$ from which it follows that $c\fus d\in F$.

We can now apply Zorn's lemma to get the existence of a maximal element of $\mathscr{P}(a,b)$, which we will call
\[
((\hat{F}_1,,\hat{I}_1),(\hat{F}_2,\hat{I}_2))
\]
and we claim that $\hat{F}_1,\hat{F}_2$ are two prime filters satisfying conditions (1)-(3) which we specified at the beginning of the proof. It is clear that $a\in \hat{F}_1$ and $b\in\hat{F}_2$, thus (1) and (2) are satisfied. For (3), assume that $a'\fus b'\notin F$ and that $a'\in \hat{F}_1$, then by construction, $b'\in\hat{I}_2$, and since $\hat{F}_2\cap \hat{I}_2=\emptyset$ we get $b'\notin\hat{F}_2$ which is what we needed to show. The last step of the proof is to show that $\hat{F}_1,\hat{F}_2$ are prime. Assume that $c\vee c'\in \hat{F}_1$ but that $c,c'\notin\hat{F}_1$. It follows that
\begin{align*}
((\hat{F}_1,,\hat{I}_1),(\hat{F}_2,\hat{I}_2)) \subsetneq((\langle \hat{F}_1\cup\{c\}\rangle, \hat{I}_1),(\hat{F}_2,\{d\mid \exists c\in\hspace{-2pt} \langle \hat{F}_1\cup\{c\}\rangle\text{ s.t. }c\fus d\notin F\}))
\end{align*}
where $\langle \hat{F}_1\cup\{c\}\rangle$ is the filter generated by $\hat{F}_1\cup\{c\}$. Since the left-hand side of the inequality is maximal, it cannot be the case that the right-hand side belongs to $\mathscr{P}(a,b)$ that is, one of the conditions 1.-6. cannot hold. Clearly 1. and 2. must hold, and 3. and 4. hold by construction, thus 5. or 6. cannot hold. In fact both conditions will not hold precisely if there exist $d\in \hat{F}_2$ and $f\in \hat{F}_1$ such that $(f\wedge c)\fus d\notin F$; that is, $f\wedge c\in \hat{I}_1$; that is, there exist $i\in I_1$ s.th.
\[
f\wedge c\leq i
\]
A completely similar reasoning shows that there must exist $f'\in \hat{F}_1$ and $i'\in \hat{I}_1$ such that
\[
f'\wedge c'\leq i'
\]
It thus follows that
\[
(f\vee f')\wedge (f\vee c')\wedge (c\vee f')\wedge (c\vee c')\leq i\vee i'
\]
Since $\hat{F}_1$ is a filter, $\hat{I}_1$ is an ideal, and we've assumed $c\vee c'\in\hat{F}_1$, we get that  $\hat{F}_1\cap \hat{I}_1\neq\emptyset$ which is a contradiction by virtue of the properties of elements of $\mathscr{P}(a,b)$. Thus either $c$ or $c'$ belongs to $\hat{F}_1$ which is thus prime as desired. A completely analogous argument shows that $\hat{F}_2$ is prime too.

\textbf{Right-to-left direction of (\ref{eq:ExLemm2}-\ref{eq:ExLemm3}):}  We show the contrapositive; that is, that if $a\lRes b\notin F$ there exists $F_1,F_2$ such that $a\in F_1$ and $b\notin F_2$. We proceed as in the case of $\fus$ by defining the set $\mathscr{P}(a,b)$ of filter-ideal pairs $((F_1,I_1),(F_2,I_2))$ such that
\begin{enumerate}
\item $\uparrow a\subseteq F_1$
\item $I_1=\{c\mid \exists d\in I_2 (c\lRes d)\in F\}$
\item $F_2=\{d\mid \exists c\in F_1 (c\lRes d)\in F\}$
\item $\downarrow a\subseteq I_2$
\item $F_1\subseteq\{c\mid\forall d\in I_2 (c\lRes d)\notin F\}$
\item $I_2\subseteq\{d\mid \forall c\in F_1 (c\lRes d) \notin F)\}$
\end{enumerate}
We make the following observations about $\mathscr{P}(a,b)$
\begin{itemize}
\item it is not empty: $((\uparrow a, \{c\mid \exists d\leq b (c\lRes d)\in F\}),(\{d\mid a\lRes d\in F\}, \downarrow b))\in \mathscr{P}(a,b)$. We need only check that conditions 5. and 6. are satisfied. Let $c$ be s.th. there exist $d\leq b$ with $c\lRes d\in F$, and assume $a\leq c$, it follows that 
\[
F\ni c\lRes d\leq a\lRes d\leq a\lRes b\notin F
\]
 a contradiction. Similarly, if there exist $d$ such that $a\lRes d\in F$ and $d\leq b$, then $F\ni a\lRes d\leq a\lRes b\notin F$, a contradiction.
\item it forms a poset under component-wise set inclusion
\item $I_1$ is an ideal: assume $c\in I_1$ and $c'\leq c$, since $\lRes$ is antitone in its first argument, we get $c\lRes d\leq c'\lRes d$ and thus $c'\lRes d\in F$. Moreover, if $c,c'\in I_1$ then there exist $d,d'\in I_2$ such that $c\lRes d,c'\lRes d'\in F$. Since $F$ is a filter, and $I_2$ is an ideal, we get $c\lRes (d\vee d'), c'\lRes (d\vee d')\in F$ and $(d\vee d')\in I_2$. If we now consider $(c\vee c')\lRes(d\vee d')$ we get by the anti-join preservation law of $\lRes$, $c\lRes(d\vee d')\wedge c'\lRes(d\vee d')$ which is a meet of elements of $F$ and thus an element of $F$. Thus $(d\vee d')$ witnesses that $c\vee c'\in I_1$
\item For completely dual reasons, $F_2$ is a filter.
\item $F_1\cap I_2=F_2\cap I_2=\emptyset$ for each $((F_1,I_1),(F_2,I_2))\in \mathscr{P}(a,b)$: assume $f\in F_1$ and $i\in I_1$ s.th. $f\leq i$, then by definition of $I_1$ there exist $d$ such that $i\lRes d\in F$ but since $\lRes$ is antitone in its first argument, this would mean $f\lRes d\in F$ which contradicts property 5. of $F_1$. Dually, assume that there exists $f\in F_2, i\in I_2$ s.th. $f\leq i$, then by definition of $F_2$, there exist $d$ such that $f\lRes d\in F$, but then we would also have $i\lRes d\in F$ which contradicts the property 6. of $I_2$.
\end{itemize}

We now check that $\mathscr{P}(a,b)$ has upper bound of chains, let $((F_1^i,I_1^i),(F_2^i,I_2^)) \in \mathscr{P}(a,b), i\in \omega$ and define
\[
F_i^\infty=\bigcup_j F_i^j, \hspace{2em}I_i^\infty=\bigcup_j I_i^j, \hspace{2em}i=1,2
\]
It is not difficult to see that $((F_1^\infty, I_1^\infty),(F_2^\infty, I_2^\infty))\in \mathscr{P}(a,b)$ by proceeding as in the existence lemma for $\fus$. We then apply Zorn's lemma to get a maximal element $((\hat{F}_1,\hat{I}_1),(\hat{F}_2,\hat{I}_2))$ of $\mathscr{P}(a,b)$. It is clear that $a\in\hat{F}_1, b\notin\hat{F}_2$. We need only check that they are prime filters. Assume $c\vee c'\in \hat{F}_1$ but $c,c'\notin \hat{F}_1$, it follows that
\[
((\hat{F}_1,\hat{I}_1),(\hat{F}_2,\hat{I}_2))\subsetneq ((\langle \hat{F}_1\cup\{c\}\rangle, \hat{I}_1),(\{c\mid \exists d\in \langle \hat{F}_1\cup\{c\}\rangle (c\lRes d)\in F\},\hat{I}_2)
\]
Since $((\hat{F}_1,\hat{I}_1),(\hat{F}_2,\hat{I}_2))$  is maximal, the right-hand side of the inequality must violate either 5. or 6., which in fact amounts to the same thing, namely the existence of $f\in \hat{F}_1, d\in \hat{I}_2$ s.th. $(f\vee c)\lRes d\in F$; that is, $(f\wedge c)\in \hat{I}_1$; that is, 
$(f\vee c)\leq i$, for some $i\in \hat{I}_1$. A similar argument implies the existence of $f'\in \hat{F}_1, i'\in \hat{I}_1$ such that $f'\wedge c\leq i'$ and a contradiction follows as in the proof for $\fus$.
To show that $\hat{F}_2$ is prime is equivalent to showing that $\hat{I}_2$ is prime, and a proof totally dual to the above proof shows just that. The proof for $\rRes$ is clearly identical.
\end{proof}

\begin{proof}[Proof of Proposition \ref{prop:JonTarskiCanExt}]
Recall Diagram (\ref{diag:jontar}):
\[
\xymatrix@C=12ex
{
\LRL A\ar[r]^{\LRL\eta_A}\ar[ddd]_{\alpha}\ar[ddr]_{\eta_{LA}} & \LRL \Upsets\pf A \ar[d]^{\semRL_{\pf A}}\\
& \Upsets \TRL \pf A \ar[d]^{\Upsets(\zeta_A)}\\
& \Upsets \pf \LRL A\ar[d]^{\Upsets\pf\alpha} \\
A\ar[r]_{\eta_A} & \Upsets \pf A
}
\]
where $\zeta_A$ is the canonical model structure map whose existence we have established in Theorem \ref{thm:rightInvDelta}. Recall that by definition of $\LRL$, $\alpha$ is equivalent to a nullary and three binary maps on $\Forg A$, which we denote as $I,\alpha_\fus,\alpha_{\ulRes}$ and $\alpha_{\urRes}$ satisfying the distribution laws DL\ref{ax:distribLaw1}-DL\ref{ax:distribLaw6}.  

Similarly, $\Upsets\pf \alpha\circ \Upsets \zeta_A\circ \semRL_{\pf A}$ is equivalent to a nullary operator and three binary maps on $\Forg \Upsets\pf A=\Forg A\ce$ which we will denote by $\Upsets\pf\alpha\circ \Upsets\zeta_A\circ I',\Upsets\pf\alpha\circ \Upsets\zeta_A\circ\delta_{\fus}, \Upsets\pf\alpha\circ \Upsets\zeta_A\circ\delta_{\ulRes}$ and $\Upsets\pf\alpha\circ \Upsets\zeta_A\circ\delta_{\urRes}$ and satisfy DL\ref{ax:distribLaw1}-DL\ref{ax:distribLaw6}. By commutativity of the above diagram these operators are extensions of $I,\alpha_\fus,\alpha_{\ulRes}$ and $\alpha_{\urRes}$ respectively. We want to show that they are in fact their unique canonical extensions. The treatment of the nullary operator is trivial. For the binary operators, we will show that they are smooth and thus equal to the unique canonical extensions $\alpha_{\fus}\ce, \alpha_{\ulRes}\ce$ and $\alpha_{\urRes}\ce$ respectively, by Proposition \ref{prop:canExtTopology}. We start by proving the following claim which readily generalises to the $n$-ary case.

\textbf{Claim}: Let $A, B$ be DLs and let $f:\Forg A\to \Forg B$. Assume that $\tilde{f}:\Forg A\ce\to \Forg B\ce$ is an extension of $f$ that (anti-)preserves all binary joins or (anti-)preserves binary meets, then $\tilde{f}$ is $(\sigma,\gamma)$-continuous and thus smooth.

\textbf{Proof}:  Since $\tilde{f}$ preserve all binary joins, its restriction $f$ preserves binary joins, and is thus smooth. Moreover, we also know 
that the canonical extension $f\ce=f\ced$ preserves all non-empty joins. If $\tilde{f}$ preserves all non-empty joins, then in particular it preserves all up-directed ones, and $f$ is consequently 
$(\gamma\eup,\gamma\eup)$-continuous. Thus we need only show that it is $(\sigma,\gamma\edown)$-continuous too. 
In fact, we show the stronger statement that $\tilde{f}$ is $(\sigma\edown,\sigma\edown)$-continuous. To see this, note first that since $\tilde{f}$ extends $f$ and preserves non-empty joins we have for every $u\in O(A)$
\begin{align*}
\tilde{f}(u)&=\tilde{f}(\bigvee\{a\in A\mid a\leq u\})\\
&=\bigvee\{\tilde{f}(a)\mid A\ni a\leq u\}\\
&=\bigvee\{f(a)\mid A\ni a\leq u\}=f\ced(u)=f\ce(u) . 
\end{align*}
That is, $\tilde{f}$ agrees with $f\ce$ on open elements. Now we use the fact that since $f\ce$ preserves non-empty joins, 
it is $(\sigma\edown, \sigma\edown)$-continuous (see Proposition \ref{prop:presPropAndTopol}); that is, that for any open $v\in O(B)$, there exists $u\in O(A)$ such that 
$f\ce(u)=v$ and thus $(f\ce)\inv(\downarrow v)=\downarrow u$. But since $\tilde{f}$ and $f\ce$ coincide on open elements, 
this also means that $\tilde{f}(u)=v$; that is, that $(\tilde{f})\inv(\downarrow v)=\downarrow u$; that is, $\tilde{f}$ is $(\sigma\edown, \sigma\edown)$-continuous.

The proof for the other preservation properties are very similar. Assume for example that $\tilde{f}$ preserves non-empty meets, it preserve down-directed ones and is thus $(\gamma\edown,\gamma\edown)$ continuous . Moreover $f\ce=f\ced$ preserves 
non-empty meets and $\tilde{f}$ agrees with $f\ce$ on all closed elements. Since $f\ce$ preserve non-empty meets, it is 
$(\sigma\eup,\sigma\eup)$-continuous, and thus so is $\tilde{f}$ by definition of $\sigma\eup$ and the fact that $\tilde{f}$ 
and $f\ce$ agree on closed elements. The proof for the anti-preservation properties are similar, with $(\sigma\edown,\sigma\eup)$ and $(\sigma\eup,\sigma\edown)$-continuity 
being shown for anti-join preservation and anti-meet preservation respectively.

This having been established, we can now return to our proof. Since $\Upsets\pf \alpha$ and $\Upsets \zeta_A$ are inverse images, they preserve any meet and any join, and in particular down-directed meets and up-directed joins. They are therefore $(\gamma,\gamma)$-continuous. Note that this concept makes sense for maps between any DLs, not just canonical extensions. All that needs to be shown now is that $\delta_{\fus},\delta_{\ulRes}$ and $\delta_{\urRes}$ have one of the preservation properties of proposition \ref{prop:presPropAndTopol}. We start with $\delta_\fus$: for any $u_i,v\in \Upsets\pf A, i\in I$ we have
\begin{align*}
\delta_\fus(\bigvee_{i\in I} u_i,v)&=\{t\in \TRL\pf A\mid \exists (x,y)\in \pi_2(t) \exists i\in I, x\in u_i, y\in v\}\\
&=\bigcup_{i\in I}\{t\in \TRL\pf A\mid \exists (x,y)\in \pi_2(t) \hspace{1ex} x\in u_i, y\in v\}\\
&=\bigvee_{i\in I}\delta_\fus(u_i,v)
\end{align*}
and similarly for the second argument. Thus $\delta_\fus$ preserves all non-empty joins, and by Proposition \ref{prop:presPropAndTopol} it is therefore $(\sigma^2,\gamma)$-continuous. Since  $\Upsets\pf \alpha$ and  $\Upsets \zeta_A$ are $(\gamma,\gamma)$-continuous, we get that $\Upsets\pf \alpha\circ \Upsets \zeta_A\circ \delta_\fus$ is $(\sigma^2,\gamma)$-continuous and thus
\[
\alpha_\fus\ce=\Upsets\pf \alpha\circ \Upsets \zeta_A\circ \delta_\fus
\]
by Proposition \ref{prop:canExtTopology}.

We can similarly show that $\delta_/$ (resp. $\delta_{\ulRes})$ preserves all non-empty meets in its first (resp. second) argument and anti-preserves 
non-empty joins in its second (resp. first) argument. As an illustration,
\begin{align*}
\delta_{\ulRes}(\bigvee_{i\in I} u_i,v)&=\{t\in  \TRL\pf A\mid\forall (x,y)\in \pi_3(t)\hspace{1ex}x\in \bigvee_{i\in I} u_i\Rightarrow y\in v\}\\
&=\{t\in  \TRL\pf A\mid\forall (x,y)\in \pi_3(t)\hspace{1ex}x\notin \bigvee_{i\in I} u_i\text{ or } y\in v\}\\
&=\bigcap_{i\in I}\{t\in  \TRL\pf A\mid\forall (x,y)\in \pi_3(t)\hspace{1ex} x\notin u_i\text{ or }y\in v\}\\
&=\bigwedge_{i\in I}\delta_{\ulRes}(u_i,v)
\end{align*}

In consequence we also get that
\[
\alpha_{\ulRes}\ce=\Upsets\pf \alpha\circ \Upsets \zeta_A\circ \delta_{\ulRes}
\]
and similarly for $\rRes$, which concludes the proof.
\end{proof}

\end{document}